\documentclass[a4paper,runningheads]{llncs}

\usepackage{etoolbox}
\usepackage[paperheight=235mm, paperwidth=155mm,textwidth=12.2cm,textheight=19.3cm,hmarginratio=1:1]{geometry}

\usepackage[T1]{fontenc}
\usepackage{amsmath,amssymb}
\usepackage{xspace}
\usepackage{booktabs,multirow}
\usepackage{siunitx} %
\usepackage{cite}
\usepackage{wrapstuff}
\usepackage{graphicx}
\usepackage{xcolor}
\usepackage[%
  bookmarks,
  unicode,
  colorlinks=true,
  allcolors=blue!65!black!90,
  breaklinks=true]{hyperref}
\usepackage{bbding}
\usepackage{mathrsfs}
\usepackage{stmaryrd}
\usepackage{adjustbox} %
\usepackage{mathtools}
\usepackage{pgfplots}
\pgfplotsset{compat=1.18}

\usepackage[inline,shortlabels]{enumitem}
\setitemize{noitemsep,topsep=0pt,parsep=0pt,partopsep=0pt,leftmargin=12.0pt}
\setenumerate{noitemsep,topsep=0pt,parsep=0pt,partopsep=0pt}
\setdescription{noitemsep,topsep=2pt,parsep=0pt,partopsep=0pt,leftmargin=12.5pt}

\usepackage{tikz}
\usetikzlibrary{calc}

\usetikzlibrary{calc,automata,positioning,decorations.pathreplacing}
\tikzset{align at top/.style={baseline=(current bounding box.north)}}
\tikzstyle{every node}=[font=\scriptsize]
\tikzstyle{state} = [draw,fill=white,circle,minimum size=4mm,inner sep=0pt,thick]
\tikzstyle{loc} = [draw,fill=white,rectangle,rounded corners,thick,align=center]
\tikzstyle{lstate} = [draw,fill=white,rectangle,rounded corners,thick,align=center,inner sep=2pt,minimum size=4.5mm]
\tikzstyle{dot} = [fill,circle,inner sep=0mm,minimum size=1.25mm,line width=0mm]
\tikzstyle{odot} = [draw,circle,inner sep=0mm,minimum size=1.25mm,line width=0mm]

\usepackage{listings}
\definecolor{codegray}{rgb}{0.5,0.5,0.5}
\definecolor{codered}{rgb}{0.67,0.0,0.0}
\definecolor{codeblue}{rgb}{0.0,0.0,0.67}
\definecolor{codegreen}{rgb}{0.0,0.45,0.0}
\makeatletter
\lst@Key{countblanklines}{true}[t]%
    {\lstKV@SetIf{#1}\lst@ifcountblanklines}

\lst@AddToHook{OnEmptyLine}{%
    \lst@ifnumberblanklines\else%
       \lst@ifcountblanklines\else%
         \advance\c@lstnumber-\@ne\relax%
       \fi%
    \fi%
		\vspace{\dimexpr-\baselineskip+\smallskipamount}}
\makeatother
\lstdefinelanguage{Modest}{
  keywords=[2]{option,action,timer,do,alt,when,stop},
  keywordstyle=[2]\color{codeblue},
  keywords=[3]{Uni,tau},
  keywordstyle=[3]\color{codered},
  keywordstyle=[1]\color{codegreen}
}

\makeatletter
\def\orcidID#1{\textsuperscript{\,\smash{\protect\raisebox{-1.25pt}{\href{http://orcid.org/#1}{\protect\includegraphics[scale=.8]{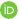}}}}}}
\makeatother

\makeatletter%
\g@addto@macro\normalsize{%
  \setlength\abovedisplayskip{3pt}%
  \setlength\belowdisplayskip{3pt}%
  \setlength\abovedisplayshortskip{-3pt}%
  \setlength\belowdisplayshortskip{3pt}%
}%
\makeatother
\allowdisplaybreaks %

\DeclareFontFamily{U}{mathx}{\hyphenchar\font45}
\DeclareFontShape{U}{mathx}{m}{n}{
      <5> <6> <7> <8> <9> <10>
      <10.95> <12> <14.4> <17.28> <20.74> <24.88>
      mathx10
      }{}
\DeclareSymbolFont{mathx}{U}{mathx}{m}{n}
\DeclareMathSymbol{\bigtimes}{1}{mathx}{"91}

\usepackage[nameinlink,capitalize,english]{cleveref}
\Crefname{figure}{Fig.}{Figs.}
\crefname{figure}{fig.}{figs.}
\Crefname{tabular}{Tab.}{Tabs.}
\crefname{tabular}{tab.}{tabs.}
\Crefname{section}{Sec.}{Sects.}
\crefname{section}{sec.}{sects.}
\Crefname{appendix}{App.}{Apps.}
\crefname{appendix}{app.}{apps.}
\Crefname{equation}{Eq.}{Eqs.}
\crefname{equation}{eq.}{eqs.}
\creflabelformat{equation}{#2#1#3}
\Crefname{example}{Ex.}{Exs.}
\crefname{example}{ex.}{exs.}
\Crefname{definition}{Def.}{Defs.}
\crefname{definition}{def.}{defs.}

\makeatletter
\AddToHook{cmd/appendix/before}{\def\cref@section@alias{appendix}}
\makeatother

\renewcommand{\emptyset}{\varnothing}
\newcommand{\RR}{\ensuremath{\mathbb{R}}\xspace}  %
\newcommand{\RRnn}{\ensuremath{\RR_{\geqslant0}}\xspace}   %
\newcommand{\RRnninf}{\ensuremath{\RR_{\geqslant0}^{\infty}}\xspace}   %
\newcommand{\RRpos}{\ensuremath{\RR_{>0}}\xspace}  %

\newcommand{\eg}{e.g.\ }
\newcommand{\ie}{i.e.\ }
\newcommand{\etal}{et al.\xspace}

\newcommand{\set}[1]{\ensuremath{\{\,#1\,\}}}
\newcommand{\tuple}[1]{\ensuremath{\langle #1 \rangle}}
\newcommand{\powerset}[1]{\ensuremath{\mathcal{P}({#1})}\xspace}
\newcommand{\defeq}{\mathrel{\vbox{\offinterlineskip\ialign{\hfil##\hfil\cr{\tiny \rm def}\cr\noalign{\kern0.30ex}$=$\cr}}}}

\newcommand{\DDists}[1]{\ensuremath{\mathit{DPD}({#1})}\xspace}
\newcommand{\Dists}[1]{\ensuremath{\mathit{PD}({#1})}\xspace}
\newcommand{\CDists}[1]{\ensuremath{\mathit{CPD}({#1})}\xspace}
\newcommand{\support}[1]{\ensuremath{\mathit{spt}({#1})}\xspace}
\newcommand{\xtr}[1]{\xrightarrow{\protect{\raisebox{-1pt}[0pt][0pt]{\ensuremath{\scriptstyle{#1}}}}}}
\newcommand{\xtrl}[1]{\xrightarrow{#1}}
\newcommand{\Pmin}[1]{\ensuremath{\mathrm{P}_\mathrm{\!min}(#1)}}
\newcommand{\Pmax}[1]{\ensuremath{\mathrm{P}_\mathrm{\!max}(#1)}}
\newcommand{\pmin}{\ensuremath{p_\mathit{min}}}
\newcommand{\pmax}{\ensuremath{p_\mathit{max}}}

\newcommand{\tool}[1]{\textsc{#1}}
\newcommand{\lang}[1]{\textsc{#1}}
\newcommand{\modest}{\lang{Modest}\xspace}
\newcommand{\jani}{\lang{Jani}\xspace}
\newcommand{\toolset}{\tool{Modest Toolset}\xspace}

\newcommand{\Timers}{\ensuremath{\mathcal{V}}\xspace}
\newcommand{\Acts}{\ensuremath{\mathcal{A}}\xspace}
\newcommand{\Locs}{\ensuremath{\mathcal{L}}\xspace}
\newcommand{\IniLoc}{\ensuremath{\ell_I}\xspace}
\newcommand{\Edges}{\ensuremath{\mathcal{E}}\xspace}

\usepackage{pifont}
\newcommand{\cmark}{\ding{51}}
\newcommand{\xmark}{\ding{55}}

\newcommand{\States}{\ensuremath{\mathcal{S}}\xspace}
\newcommand{\InitialState}{\ensuremath{s_I}\xspace}
\newcommand{\Trans}{\ensuremath{\mathcal{T}}\xspace}

\newcommand{\sem}[1]{\ensuremath{\llbracket{#1}\rrbracket}}
\newcommand{\sched}{\ensuremath{\mathfrak{s}}\xspace} %

\newcommand{\OCInt}{\ensuremath{\mathfrak{I}}}

\newcommand{\powersetP}{\ensuremath{\mathscr{P}}}
\newcommand{\SOCInt}{\ensuremath{\powersetP_{\mathrm{nol}}(\OCInt)}}
\newcommand{\Part}{\ensuremath{\mathfrak{p}}}
\newcommand{\da}{\ensuremath{\mathsf{ia}}}
\newcommand{\soj}{\sigma}
\newcommand{\expired}{\lightning}

\newcommand{\dedrule}[2]{\frac{#1}{#2}}
\newcommand{\restart}{R} %

\begin{document}

\title{%
Effective Stochastic Automata Model\\ Checking by Interval Abstraction\thanks{
Authors are sorted in alphabetical order.
This work was supported
by the EU's H2020 R\,\&\,I programme under MSCA grant agreement 101008233 (MISSION),
by the Interreg North Sea project STORM\_SAFE,
by SeCyT-UNC grant 33620230100384CB (MECANO),
and
by NWO VIDI grant VI.Vidi.223.110 (TruSTy).
}\\ (extended version)}
\titlerunning{Stochastic Automata Model Checking by Interval Abstraction}

\author{%
Pedro R.\ D'Argenio\inst{1}\orcidID{0000-0002-8528-9215}
\and
Arnd Hartmanns\inst{2}\orcidID{0000-0003-3268-8674}
\and
Annabell Petri\inst{2}$^{\text{\,(\raisebox{-1.6pt}{\Envelope})}}$%
}
\institute{%
Universidad Nacional de Córdoba and CONICET, Córdoba, Argentina
\and
University of Twente, Enschede, The Netherlands
$\cdot$ \email{annabell.petri@utwente.nl}
}
\authorrunning{P.\ R.\ D'Argenio, A.\ Hartmanns, A.\ Petri}

\maketitle

\begin{abstract}
Stochastic automata (SA) are a formal stochastic continuous-time model based on countdown timers whose expiration times follow general probability distributions.
SA are particularly useful to faithfully model and analyse dependable systems involving faults, maintenance, and repairs.
Effective SA analysis approaches have so far been limited to statistical model checking and thus deterministic SA, while previously proposed model-checking techniques apply to limited subclasses of SA only, or do not scale.
In this paper, we present the first dedicated SA model checking approach that is general and effective:
It puts few restrictions on the input SA, and we show in our experimental evaluation that it works well for nontrivial examples.
It combines a refinable interval abstraction of the continuous distributions with a direct application of the ``big time steps'' semantics of SA, providing upper/lower bounds on maximum/minimum reachability probabilities.
We extend the \modest and \jani modelling formalisms with support for SA, and provide a prototype implementation of our approach in Rust.
\end{abstract}

\section{Introduction}
\label{sec:Intro}

In dependable systems, component faults happen randomly over time, with schemes involving redundancy, regular inspections, maintenance, and repairs attempting to prevent escalation into overall system failure.
Similarly, in high-performance systems, random queueing and service times, message loss probabilities, and signal transmission and propagation delays determine the system's key performance properties such as throughput or response times.
Formal models for dependability and performance evaluation~\cite{BHHK10} need to capture these aspects of stochastic time in a realistic way, as close as possible to the real system or the available data.
Two widely-used formalisms are continuous-time Markov chains (CTMCs) and generalized stochastic Petri nets (GSPNs)~\cite{MCB84,EHKZ13}.
In both, all delays must be exponentially distributed.
The resulting memoryless nature of these models admits scalable analytical analysis methods such as probabilistic model checking~\cite{BAFK18,BHHK03}.
In reality, however, many inter-event times are not exponentially distributed, and hard to approximate by CTMCs/phase-type distributions:
time to failure is more realistically modelled by a Weibull distribution in most cases, and inspections typically happen around fixed intervals.
We thus need \emph{non-Markovian} formalisms that directly incorporate general continuous probability distributions, such as stochastic automata (SA)~\cite{DK05} or Petri nets with general transitions (\eg HPnGs~\cite{GRH14}).
In this paper, we use SA, because they are conceptually simple yet highly expressive.
Additionally, their compositionality makes them attractive for modelling complex component-based systems.

SA extend labelled transition systems with stochastic \emph{timers} (that may be \emph{reset} to values sampled from continuous probability distributions to then decrease over time and \emph{expire} when reaching value zero) and \emph{guard sets} that enable an edge when all their timers have expired.
As SA are non-Markovian \emph{and} allow nondeterministic choices, their analysis is hard.
The only available tool with dedicated SA support, \tool{Fig}~\cite{Bud22}, employs statistical model checking~\cite{LLTYSG19,AP18} (\ie Monte Carlo simulation) and is thus restricted to (weakly) deterministic SA~\cite{DM18}.
Analytical approaches for more general models, such as HPnGs or stochastic timed automata (STA)~\cite{BDHK06,HHH14}, do not scale for typical SA that have many discrete locations and several timers possibly reset repeatedly on loops.

We present the first dedicated model checking approach for SA that is effective (\ie works for typical and nontrivial SA models) and implemented in an available tool.
It is based on \emph{interval abstraction}~\cite{FHHWZ11}:
we replace each clock reset, \ie each sampling from a continuous distribution $\mu$, by a discrete distribution over a finite set of intervals that partition $\mu$'s support.
The choice of concrete value from an interval is nondeterministic, and we propagate the intervals symbolically through the semantics of the SA.
This turns the SA into a Markov decision process (MDP)~\cite{Bel57,How60} that overapproximates the SA's semantics.
As a result, the max.\ (min.) probability of eventually reaching a set of MDP states corresponding to a goal location in the SA---which we compute by value iteration~\cite{HJQW23}---is an upper (lower) bound on the corresponding probability in~the~SA.

We specify our approach formally and show the overapproximation property (\Cref{sec:OurApproach}).
We extend the \modest~\cite{BDHK06,HHHK13} and \jani~\cite{BDHHJT17} languages with support for timers and thus SA, providing a new user-friendly high-level modelling language for SA and a means to conveniently exchange SA models between tools (\Cref{sec:Languages}).
Our new SA model checking approach is implemented in a prototype tool written in Rust, which we use to evaluate its performance and scalability characteristics on example SA from and inspired by the literature (\Cref{sec:Experiments}).
We find that, despite its prototypical nature, our tool already works well for reasonably-sized SA even when using many intervals to obtain results with good accuracy.

\paragraph{Related work.}
\tool{Fig}'s simulation-based approach (mentioned above) is limited to a subclass of \emph{input-output} SA that syntactically guarantees all nondeterminism to be spurious, \ie minimum and maximum reachability probabilities coincide.

Two SA model checking algorithms were proposed by Bryans \etal in 2003~\cite{BBD03}.
They use a clock-based variant of SA with severe restrictions (clocks may only be used out of locations in which they have just been reset, and all distributions must be non-negative and bounded). %
The first algorithm uses a ``region tree'' and requires solving nested integrals over the probability density functions of each expired clock, which becomes infeasible as the tree grows in depth.
The second algorithm avoids the nested integrals by discretising the distributions.
No tool implements either algorithm today.
We note that our interval abstraction does not discretise the sampling results, but instead replaces continuous probabilistic choices by discrete probabilistic followed by continuous nondeterministic ones.
For nondeterministic SA, both algorithms require a strategy to be user-specified, while our approach includes the optimisation over nondeterministic choices.

Model checking tool support exists for the following stochastic-timed formalisms that are more general than SA:
(1)~STA extend timed automata by allowing clocks to be compared to real-valued variables whose values can be sampled from general probability distributions.
Compared to SA, they also support continuous nondeterministic delays, and are very (unrealistically) flexible in how clocks and the real-valued variables can be combined.
Interval abstraction has been implemented for STA in the \toolset's~\cite{HH14} \tool{mcsta} model checker~\cite{HHH14} in a way that reduces the STA to a probabilistic timed automaton~\cite{KNSS02}, which is turned into an MDP via the digital clocks approach~\cite{HMP92,KNPS06}.
The implementation is limited to abstraction intervals of width~1 (or infinity).
The digital clocks semantics enumerates all time steps of duration~1, making the MDP very large, but also preventing refinement to smaller intervals.
(2)~Stochastic hybrid automata (SHA)~\cite{FHHWZ11} further extend STA with continuous variables whose evolution over time is governed by differential equations and inclusions.
The \tool{prohver} tool applies interval abstraction to SHA, followed by a hybrid reachability analysis via \tool{Phaver}~\cite{Fre08}, to overapproximate reachability probabilities.
Here, the abstraction intervals have to be specified manually by the user, and the hybrid reachability step---unnecessary for SA---is the scalability and performance bottleneck.
(3)~While \tool{prohver} applies to linear SHA of type I, the \tool{Realyst} tool~\cite{DSSR23} focuses on rectangular hybrid automata with random events (RAEs), into which (some) SA can also be encoded.
\tool{Realyst}'s approach works for time-bounded reachability only, and requires structurally non-Zeno RAEs, which complicates the analysis of SA with cycles.
As every timer reset must be encoded as a separate random variable in the RAE, each of which \tool{Realyst} turns into one dimension of the model's continuous state, scalability---already limited by the complexity of hybrid reachability---for typical SA is low.

Finally, Buchholz \etal~\cite{BKS14} present model checking algorithms for a generalization of CTMCs with clocks that they also call ``stochastic automata'', but clearly point out that their definition is less general than the one of~\cite{DK05} that~we~use.

\section{Background}
\label{sec:Background}

Let $\RRnn$ be $[0, \infty)$, %
$\RRpos \defeq (0, \infty)$, and $\RR_\infty \defeq \RR \cup \set{\infty}$.
A discrete \emph{probability distribution} is a function $\mu\colon S \to [0,1]$ over a countable set $S$ with $\sum_{s \in S}\mu(s) = 1$.
Its \emph{support} is $\support{\mu} \defeq \set{ s \in S | \mu(s) > 0}$.
Now let $S$ be an arbitrary set.
We write $\powerset{S}$ for its powerset.
Given a $\sigma$-algebra $\sigma(S)$ over $S$, subset $A \subseteq S$ is \emph{measurable} if $A \in \sigma(S)$.
A \emph{probability measure} over $S$ is a function $\mu\colon \sigma(S) \to [0,1]$ with $\mu(S) = 1$ and $\mu(\bigcup_{i \in I}A_i) = \sum_{i \in I} \mu(A_i)$ for any countable index set $I$ and measurable pairwise disjoint $A_i \subseteq S$.
The \emph{support} of probability measure $\mu$ is the smallest compact (measurable) set $\support{\mu} \subseteq S$ s.t.\ $\mu(\support{\mu})=1$.
Given two probability measures $\mu_1$ and $\mu_2$ over $S$, the \emph{product measure} is the unique probability measure such that $(\mu_1 \otimes \mu_2) (A_1 \times A_2) = \mu_1(A_1) \cdot \mu_2 (A_2)$ for all measurable $A_1, A_2 \subseteq S$.
A probability measure $\mu$ over $\RR$ is absolutely continuous if $\mu(A) = 0$ for all $A$ of Lebesgue measure~$0$; we call it a \emph{continuous probability distribution}~(CPD).
$\DDists{S}$ is the set of all discrete probability distributions, $\Dists{S}$ of all probability measures, and $\CDists{S}$ of all CPDs over appropriate~$S$.

\paragraph{Stochastic automata}%
\cite{Dar99,DK05}
are labelled transition systems equipped with \emph{timers}:
real-valued variables that can be reset to values sampled from CPDs, decrease synchronously with time, and guard the SA's edges.
Formally:

\begin{definition}
\label{def:sa}
A \emph{stochastic automaton} (SA) is a tuple
$\tuple{\Locs, \IniLoc, \Acts, \Timers, \Edges}$ consisting of
finite sets of
locations \Locs with initial location $\IniLoc \in \Locs$,
action labels \Acts with $\Acts \cap \RR = \varnothing$,
timer variables \Timers,
and
edges $\Edges \subseteq \Locs \times \powerset{\Timers} \times \Acts \times \powerset{\Timers} \times \Locs$.
Each timer $x$ has an associated CPD $\mu_x$ with $\mu_x(\RRnn) = 1$.
We write $\ell \xtr{G, a, R} \ell'$ for $\tuple{\ell, G, a, R, \ell'} \in \Edges$; such an edge consists of
a source location~$\ell$,
a guard set $G \subseteq \Timers$,
an action $a \in \Acts$,
a reset set $R \subseteq \Timers$,
and
a target location~$\ell'$.
\end{definition}
Unless noted otherwise, we assume SA $\mathcal{M} = \tuple{\Locs, \IniLoc, \Acts, \Timers, \Edges}$ to be given in the remainder of this paper.
Let $\vec{v}\colon \Timers \to \RR$ be a timer valuation.
Let valuation $\vec{0}$ be defined by, for all $t \in \Timers$, $\vec{0}(t) = 0$, and $\vec{v} - d$ for $d \in \RR$ by $(\vec{v} - d)(t) = v(t) - d$. %

In an SA's semantics,
we start in $\IniLoc$ with all timers having value~$0$.
When in location $\ell$, all timers synchronously decrease at rate $1$ over time.
If $\ell \xtr{G, a, R} \ell'$ and all timers in $g$ have \emph{expired} (\ie have value $\leq 0$), then no further time can pass (maximal progress), and the edge must be taken:
the timers in $R$ are \emph{reset} (\ie assigned a new value randomly sampled from their associated CPD), and we continue in location~$\ell'$.
The semantics also allows \emph{discrete nondeterministic} choices between edges if there is another edge whose guard's timers have all expired.
Formally, it is given in terms of an infinite-state but finitely branching transition system:

\begin{definition}
\label{def:TPTS}
A \emph{timed probabilistic transition system} (TPTS) is defined as a tuple
$\tuple{\States, \InitialState, \Acts', \Trans}$
consisting of
a measurable set of states $\States$ with initial state $\InitialState \in \States$,
an alphabet $\Acts' = \RRpos \uplus \Acts$ partitioned into \emph{delays} in \RRpos and actions in $\Acts$,
and
transitions $\Trans \subseteq \States \times \Acts' \times \Dists{\States}$.
We write $s \xtr{\alpha} \mu$ for $\tuple{s, \alpha, \mu} \in \Trans$; such a transition consists of
a source state $s$, label $\alpha \in \Acts'$, and measure over target states $\mu$.
For all $s \in \States$, we require
$\Trans_s \defeq \set{ \tuple{s, \alpha, \mu} \in \Trans }$ to be finite
and
$|\Trans_s| = 1$ if $\exists\,\tuple{s, d, \mu} \in \Trans \colon d \in \RRpos$ (\ie states admitting delays are deterministic).
\end{definition}

\begin{definition}
\label{def:SASemantics}
The TPTS
$\sem{M} \defeq \tuple{\Locs \times (\Timers \to \RR_\infty), \tuple{\IniLoc, \vec{0}}, \Acts \cup \RRpos, \Trans_{\sem{M}}}$
is the \emph{semantics of an SA} $\mathcal{M}$ with
$\Trans_{\sem{M}}$ being the smallest relation satisfying
\begin{gather*}
\dedrule{
	\ell \xtrl{G, a, R} \ell' \quad \bigwedge_{x_i \in G}{\vec{v}(x_i)\leq 0}
}{
	\tuple{\ell, \vec{v}} \xtrl{a} \set{s' \mapsto 1} \otimes \left(\bigotimes_{x_i \in R}\mu_{x_i}\right) \otimes \left(\bigotimes_{x_i \notin R} \set{\vec{v}(x_i) \mapsto 1}\right)
}
\quad(\mathrm{S1})
\\[2pt]
\text{and}\quad
\dedrule{
	0 < d = \min\set{\max\set{ \vec{v}(x_i)\mid x_i \in G } \mid \ell \xtrl{G, a, R} \ell'}
}{
	\tuple{\ell, \vec{v}} \xtrl{d} \set{\tuple{\ell, \vec{v} - d} \mapsto 1}
}
\quad(\mathrm{S2})
\end{gather*}
with all resulting deadlock states receiving an $\infty$-labelled self-loop.
\end{definition}
This definition is the ``big time steps'' semantics of SA~\cite{DGHS18}:
Rule S2 admits only the delay that directly enables an edge, and no intermediate delays as in other ``dense time steps'' semantics.
For the \emph{closed} SA and properties we use, the two semantics are equivalent (in essence, weakly bisimilar in a time-abstract way).

\begin{wrapstuff}[type=figure,width=3.85cm]
	\centering
	\begin{tikzpicture}[on grid,auto]
		\node[state] (l0) {$\ell_0$};
		\coordinate[left=0.3 of l0.west] (start);
		\node[] (distr) [above right=0.2 and 1.3 of l0,align=left] {$x\colon \textsc{Uni}(2, 4)$\\$y\colon \textsc{Uni}(0, 5)$};
		\node[state] (l1) [below=1.2 of l0] {$\ell_1$};
		\node[state] (l2) [below left=1 and 0.75 of l1] {$\ell_2$};
		\node[state] (l3) [below right=1 and 0.75 of l1] {$\ell_3$};
		;
		\path[->]
		(start) edge node {} (l0)
		(l0) edge node[right,pos=0.3,inner sep=0.5mm] {\strut$\varnothing, \texttt{a},$} node[right,pos=0.7,inner sep=0.5mm] {\strut$\text{\restart}(\{ x, y \})$} (l1)
		(l1) edge[] node[left,pos=0.15,inner sep=1mm] {\strut$\{y\}, \texttt{b}$} %
		(l2)
		(l1) edge[] node[right,pos=0.15,inner sep=1mm] {\strut$\{x\}, \texttt{c}$} %
		(l3)
		(start) edge node {} (l0)
		(l2) edge[out=150,in=-145] %
		 node[left,pos=0.5,inner sep=0.5mm] {\strut$\varnothing, \tau$} (l0)	
		;
	\end{tikzpicture}
	\caption{Example SA $\mathcal{M}_\mathit{re}$.}
	\label{fig:sa_running_example}
\end{wrapstuff}

\begin{example}
\label{ex:RunningSA}
\hspace{-4.5ex} %
\Cref{fig:sa_running_example} shows our running example SA $\mathcal{M}_\mathit{re}$.
It has two timers $x$ and $y$ whose expiration times are sampled from the given uniform distributions.
The two edges labelled \texttt{a} and $\tau$ out of locations $\ell_0$ and $\ell_2$ have empty guard sets:
they must be taken immediately upon entering those locations.
In $\ell_1$, the outgoing edges can only and must be taken when the respective timer has expired.
If $y$ expires before $x$, we go to $\ell_2$, and to $\ell_3$ otherwise.
Timers are only reset on the edge from $\ell_0$ to $\ell_1$ as denoted by the reset set $\text{\restart}(\{ x, y \})$.
We omit empty~reset~sets.
\end{example}
In general, SA are \emph{compositional}:
They can be equipped with a notion of \emph{parallel composition} and an \emph{open} semantics with dense time steps and without maximal progress.
This facilitates elegant and compact models of systems consisting of concurrent components.
For analysis, usually the composition's closed semantics which excludes parallel composition and synchronization on shared actions between SA is considered, and thus we use single closed SA in this paper.

\paragraph{Properties.}
Given an SA, we are interested in the probability of reaching any of a set of goal locations $\mathcal{G} \subseteq \Locs$, defined via the probability of reaching any state in $\sem{\mathcal{G}} \defeq \set{ \tuple{\ell, \vec{v}} \mid \ell \in \mathcal{G} }$ in the TPTS semantics.
In a TPTS, a \emph{path} is an infinite sequence $s_0 \tuple{a_0, \mu_0} s_1 \tuple{a_1, \mu_1} \!\ldots \in (\States \times (\Acts' \times \Dists{\States}))^\omega$ with $s_0 = \InitialState$ and $\forall i\colon s_i \xtr{a_i} \mu_i \wedge s_{i+1} \in \support{\mu_i}$.
Paths resolve all nondeterministic and probabilistic choices.
A \emph{strategy} resolves only the nondeterminism:

\begin{definition}
\label{def:Strategy}
$\sched \colon \States \to \Acts' \times \Dists{\States}$ is a \emph{strategy} if $\sched(s) = \tuple{\alpha, \mu} \Rightarrow s \xtr{\alpha} \mu$.
\end{definition}
We work with memoryless deterministic strategies.
Randomised ones would be needed for \eg multi-objective properties~\cite{FKNPQ11},
and for reachability properties on SA, memoryless and history-dependent strategies are equally expressive~\cite{DGHS18}.

Every strategy $\sched$ induces a probability measure $\mathbb{P}_{\!\sched}$ on the space of all paths.
For a formal definition, see~\cite{Wol12}.
As is usual, we restrict to \emph{non-Zeno} strategies: %
we require $\mathbb{P}_{\!\sched}(\Pi_\infty) = 1$, where $\Pi_\infty$ is the set of paths whose sum of delays is~$\infty$.

\begin{definition}
\label{def:ReachProb}
Consider a TPTS that is the semantics of an SA with goal set $\mathcal{G} \subseteq \Locs$.
Then $\Pmin{\diamond\,\mathcal{G}}$ and $\Pmax{\diamond\,\mathcal{G}}$ are the minimum and maximum \emph{reachability probabilities} for $\mathcal{G}$, defined as $\Pmin{\diamond\,\mathcal{G}} = \inf_{\sched} \mathbb{P}_{\!\sched}(\lozenge_\mathcal{G})$ and $\Pmax{\diamond\,\mathcal{G}} = \sup_{\sched} \mathbb{P}_{\!\sched}(\lozenge_\mathcal{G})$, respectively, with $\lozenge_\mathcal{G}$ the set of paths that contain a state in $\sem{\mathcal{G}}$.
\end{definition}

\paragraph{Markov decision processes}
are finite TPTSs whose transitions lead into \emph{discrete} probability distributions.
They are the result of interval abstraction of SA.

\begin{definition}
\label{def:MDP}
A \emph{Markov decision process} (MDP) is a tuple
$\tuple{\States, \InitialState, \Acts, \Trans}$
with finite sets of
states $\States$ with initial state $\InitialState \in \States$,
action labels $\Acts$,
and
transitions $\Trans \subseteq \States \times \Acts \times \DDists{\States}$.
For all $s \in \States$, we require
non-empty $\Trans_s \defeq \set{ \tuple{s, \alpha, \mu} \in \Trans }$. %
\end{definition}
We can define reachability probabilities for MDPs in a similar way as for TPTSs.
On MDPs, they can be computed or approximated efficiently via linear programming, policy iteration, or various variants of value iteration (VI)~\cite{HJQW23}.

\section{Big-Steps Interval Abstraction for SA}
\label{sec:OurApproach}

Interval abstraction of models with continuous probability distributions was introduced for the analysis of SHA~\cite{FHHWZ11}, where it turned the SHA into a \emph{probabilistic} hybrid automaton (PHA)~\cite{Spr00,ZSRHH10}, which could in turn be analysed by combining a hybrid reachability checker with MDP model checking.
Interval abstraction replaces sampling from a CPD $\mu$ by a combination of
(1)~discrete sampling from a finite set of intervals that partition (or, in general, cover) $\support{\mu}$, with interval $I_i$ chosen with probability $\mu(I_i)$, followed by
(2)~a continuous nondeterministic choice to choose a concrete value from the sampled interval.
Interval abstraction does not discretise the final outcome:
Every value in $\support{\mu}$ can still be chosen, but instead of being fully random, the strategy now makes the choice within the sampled interval.
Since we use absolutely continuous distributions, we can neglect singleton overlaps between the intervals (as they have probability~$0$), and will thus use closed intervals throughout.
The intervals to use for the abstraction can be specified in different ways.
In our implementation and examples, we specify the desired probability mass of each interval, but we could equally specify the interval width or just the intervals themselves.

\begin{example}
\label{ex:Intervals}
When abstracting uniform distribution $\textsc{Uni}(0, 5)$ from \Cref{ex:RunningSA} with per-interval probability of $0.5$, we obtain intervals $[0, 2.5]$ and $[2.5, 5]$.
For $\textsc{Exp}(1)$, the exponential distribution with rate~$1$, we obtain $[0, -{\ln 0.5}]$ and $[-{\ln 0.5}, \infty)$.
\end{example}
We apply interval abstraction to the semantics of \Cref{def:SASemantics}, but then propagate the intervals symbolically. %
We thus avoid (in effect: postpone) the explicit continuous nondeterministic choices and obtain a finite MDP.
In this section, we specify our transformation formally and establish its relation to the TPTS semantics.

\subsection{SA Interval Abstraction}
\label{sec:OurIntervalAbstraction}

Let $\OCInt$ denote the set of all closed intervals in $\RRnninf$ together with all intervals of the form $[l,\infty)$.
Given interval $I$, let $lI=\inf(I)$ and $uI=\sup(I)$.
A set of intervals $B\subseteq\OCInt$ is \emph{non-overlapping} if for all
$I_1,I_2\in B$, $I_1\cap I_2\neq\emptyset$ implies either $uI_1=lI_2$ or $uI_2=lI_1$.  Let
$\SOCInt$ be the set of all non-overlapping subsets of $\OCInt$.
A \emph{timer-domain partition} is a function $\Part\colon \Timers \to \SOCInt$
s.t.\ for all $x\in\Timers$, $\bigcup {\Part(x)} = \support{\mu_x}$.%

Given an SA $\mathcal{M}$ as usual and a timer-domain partition
$\Part$, we define its \emph{interval abstraction} as the MDP
$\da(\mathcal{M})=(S_\da, \tuple{ \IniLoc, \vec{0}_\da }, \Acts \cup \set{\soj}, \Trans_\da)$ where
\begin{itemize}
\item%
$\soj\notin\Acts$ is a fresh new label indicating some
unspecified sojourn time;
\item%
$S_\da = \Locs \times (\Timers \to \OCInt \cup \set{\expired})$ where, for each state
$\tuple{\ell, b} \in S_\da$ and $x \in \Timers$, if $b(x)=\expired$, timer $x$ has
already expired in this state and, otherwise, $b(x)$ is the interval
containing all possible values timer $x$ can take in this state with
$lb(x)$ and $ub(x)$ denoting the interval's lower and upper bound,  respectively, and for a set $C\subseteq\Timers$ of timers, we let
$lb(C)=\max\set{lb(x)\mid x\in C \wedge b(x)\neq\expired}$ (the delay after which \emph{all} timers \emph{can} have expired) and
$ub(C)=\max\set{ub(x)\mid x\in C \wedge b(x)\neq\expired}$ (the delay after which \emph{all} timers \emph{must} have expired);
\item
$\vec{0}_\da$ maps every timer to $\expired$ in the initial state;
and
\item%
$\Trans_\da$ is the smallest relation satisfying the following two inference rules:
\[
\dedrule{
	\ell \xtrl{G, a, R} \ell'
	\quad
	\forall x \in G\colon b(x) = \expired
}{
	\tuple{\ell,b} \xtrl{a} \pi^b_{R, \ell'}
}
\qquad(\text{R1})
\]
with the transitions' target distributions defined by
\[
\pi^b_{R, \ell'}(\tuple{\ell'', b'})=
\begin{cases}
	\displaystyle\prod_{x \in R}\mu_x(b'(x))
		& \text{if }
		\begin{array}[t]{l}
			\ell''=\ell',\\
			\forall{x \notin R}\colon b'(x)=b(x), \text{ and}\\
			\forall{x \in R}\colon b'(x)\in\Part(x)
		\end{array} \\
	0 & \text{otherwise}
\end{cases}
\]
and
\[
\dedrule{
	\ell \xtrl{G, a, R} \ell'
	\quad
	\forall \ell \xtrl{G'\!\!,\,a'\!\!,\,R'} \ell'' \colon
	(\exists\, x \in G'\colon b(x) \neq \expired)
	\wedge
	lb(G) \leq ub(G')
}{
	\tuple{\ell, b} \xtrl{\soj} \set{ \tuple{\ell, d_{\ell, G}^b} \mapsto 1 }
}
\qquad(\text{R2})
\]
with
\[d_{\ell, G}^b(x)=
\begin{cases}
	\expired &\hfill\makebox[0em][r]{if $x \in G \vee b(x)=\expired \vee ub(x) \leq lb(G)$}\\[1ex]
	\Big[
		\max\{ 0, lb(x) {-} \min\{ ub(G') \mid \ell \xtrl{G'\!\!,\,a'\!\!,\,R'} \ell'' \}\},
		ub(x) {-} lb(G)
	\Big]
	&\text{else}
\end{cases}
\]
with all resulting deadlock states receiving a $\soj$-labelled self-loop.
\end{itemize}
Rule R1 creates an MDP transition to $\pi^b_{R, \ell'}$ whenever an edge in the SA would be enabled in the current state;
the discrete distribution $\pi^b_{R, \ell'}$ samples an interval for every timer that is reset.
Rule R2 creates $\soj$-labelled transitions, which represent the symbolic passage of some time, as follows:
\begin{itemize}
\item
The rule picks an edge $e = \ell \xtr{G, a, R} \ell'$ out of the current location $\ell$ in the SA that should become enabled.
If $\ell$ has no edges, then no symbolic time passes.
\item
First, \emph{no} edge must be enabled in the current state---otherwise, by maximal progress, it would have to be taken immediately.
This is condition ``$\exists\, x \in G'\colon b(x) \neq \expired$'':
In every guard set, at least one timer is not expired.
\item
Additionally, no edge may be \emph{forced} to be taken earlier.
This is condition ``$lb(G) \leq ub(G')$'':
else, edge $e'$ with guard $G'$
\emph{must} be taken after waiting for $ub(G')$ time units,
while $e$ needs \emph{at least} delay $lb(G) > ub(G')$ to be enabled.
\item
If these conditions hold, then all timers in $G$ \emph{can} expire, and we let them do so.
Function $d_{\ell, G}^b$ determines the new intervals of values that all timers not in $G$ can then be in:
The first case states that all timers are expired ($\expired$) that
(i)~are in $G$,
(ii)~have previously expired, or
(iii)~whose upper bound is at most $lb(G)$ (which applies to timers not in~$G$ that have ``low'' values).
In the second case, we have an unexpired timer $x \notin G$ that is not \emph{forced} to expire now.
Its lower bound is reduced by the maximum amount of time that could have passed, which is the lowest upper bound of the guard set of any edge out of $\ell$, but not below~$0$.
If the new lower bound is $0$, $x$ \emph{could} have expired, but since $x \notin G$, it by choice \emph{did not} expire.\!\footnote{Two timers expiring simultaneously has probability $0$ because of our restriction to absolutely continuous distributions, so we can disregard this case here.}
The upper bound, in turn, is decreased by the lowest possible delay needed for all timers in $G$ to expire.
\end{itemize}

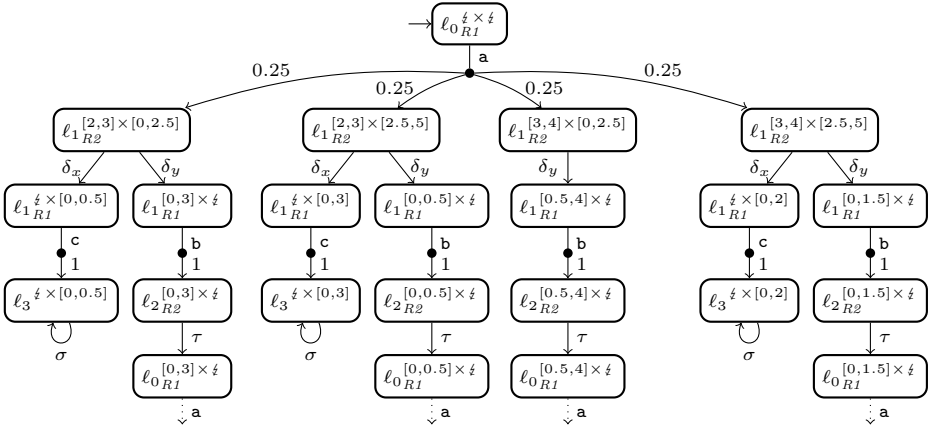
\begin{figure}[t]
	\centering
	\begin{tikzpicture}[on grid,auto]
		\node[loc] (l0) {${\ell_0}^{\expired \times \expired}_\mathit{R1}$};
		\coordinate[left=0.3 of l0.west] (start);
		\node[dot] (dot0) [below=0.65 of l0] {};
		\node[loc] (l11) [below left=0.75 and 4.6 of dot0] {${\ell_1}^{[2,3] \times [0,2.5]}_\mathit{R2}$};
		\node[loc] (l12) [below left=0.75 and 1.3 of dot0] {${\ell_1}^{[2,3] \times [2.5,5]}_\mathit{R2}$};
		\node[loc] (l13) [below left=0.75 and -1.3 of dot0] {${\ell_1}^{[3,4] \times [0,2.5]}_\mathit{R2}$};
		\node[loc] (l14) [below left=0.75 and -4.5 of dot0] {${\ell_1}^{[3,4] \times [2.5,5]}_\mathit{R2}$};

        \path[->]
		(start) edge[] (l0)
		(dot0) edge[bend right=10] node[above,pos=0.7] {$0.25$} (l11)
		(dot0) edge[bend right=10] node[left] {$0.25~$} (l12)
		(dot0) edge[bend left=10] node[right] {$~0.25$} (l13)
		(dot0) edge[bend left=10] node[above,pos=0.7] {$0.25$} (l14)
        ;

		\node[loc] (l1-1P2-1) [below left=1.0 and 0.8 of l11] {${\ell_1}^{\expired \times [0,0.5]}_\mathit{R1}$};
		\node[loc] (l1-1P2-2) [below left=1.0 and -0.8 of l11] {${\ell_1}^{[0,3] \times \expired}_\mathit{R1}$};
		\node[loc] (l1-2P2-1) [below left=1.0 and 0.8 of l12] {${\ell_1}^{\expired \times [0,3]}_\mathit{R1}$};
		\node[loc] (l1-2P2-2) [below left=1.0 and -0.8 of l12] {${\ell_1}^{[0,0.5] \times \expired}_\mathit{R1}$};
		\node[loc] (l1-3P2-1) [below left=1.0 and 0 of l13] {${\ell_1}^{[0.5,4]\times \expired}_\mathit{R1}$};
		\node[loc] (l1-4P2-1) [below left=1.0 and -0.8 of l14] {${\ell_1}^{[0,1.5] \times \expired}_\mathit{R1}$};
		\node[loc] (l1-4P2-2) [below left=1.0 and 0.8 of l14] {${\ell_1}^{\expired \times [0,2]}_\mathit{R1}$};

        \path[->]
        (l11) edge node[left] {$\delta_x$} (l1-1P2-1)
		(l11) edge node[right] {$\delta_y$} (l1-1P2-2)
		(l12) edge node[left] {$\delta_x$} (l1-2P2-1)
		(l12) edge node[right] {$\delta_y$} (l1-2P2-2)
		(l13) edge node[left] {$\delta_y$} (l1-3P2-1)
		(l14) edge node[right] {$\delta_y$} (l1-4P2-1)
		(l14) edge node[left] {$\delta_x$} (l1-4P2-2)
        ;

        \node[loc] (l3-3) [below left=1.25 and 0 of l1-1P2-1] {${\ell_3}^{\expired \times [0,0.5]}_\mathit{\phantom{R}}$};
		\node[loc] (l2-4) [below left=1.25 and 0 of l1-1P2-2] {${\ell_2}^{[0,3] \times \expired}_\mathit{R2}$};
		\node[loc] (l3-1) [below left=1.25 and 0 of l1-2P2-1] {${\ell_3}^{\expired \times [0,3]}_\mathit{\phantom{R}}$};
		\node[loc] (l2-1) [below left=1.25 and 0 of l1-2P2-2] {${\ell_2}^{[0,0.5] \times \expired}_\mathit{R2}$};
		\node[loc] (l2-2) [below left=1.25 and 0 of l1-3P2-1] {${\ell_2}^{[0.5,4] \times \expired}_\mathit{R2}$};
		\node[loc] (l3-2) [below left=1.25 and 0 of l1-4P2-2] {${\ell_3}^{\expired \times [0,2]}_\mathit{\phantom{R}}$};
		\node[loc] (l2-3) [below left=1.25 and 0 of l1-4P2-1] {${\ell_2}^{[0,1.5] \times \expired}_\mathit{R2}$};

		\node[loc] (l0-1) [below=1.0 of l2-1] {${\ell_0}^{[0,0.5] \times \expired}_\mathit{R1}$};
		\node[loc] (l0-2) [below=1.0 of l2-2] {${\ell_0}^{[0.5,4] \times \expired}_\mathit{R1}$};
		\node[loc] (l0-3) [below=1.0 of l2-3] {${\ell_0}^{[0,1.5] \times \expired}_\mathit{R1}$};
		\node[loc] (l0-4) [below=1.0 of l2-4] {${\ell_0}^{[0,3] \times \expired}_\mathit{R1}$};

        \node[] (dot1) [below=0.75 of l0-1] {};
		\node[] (dot2) [below=0.75 of l0-2] {};
		\node[] (dot3) [below=0.75 of l0-3] {};
		\node[] (dot4) [below=0.75 of l0-4] {};

		;
		\path[]
		(l0) edge node[pos=0.55] {\texttt{a}} (dot0)
		(l0-1) edge[dotted, ->] node[pos=0.55] {\texttt{a}} (dot1)
		(l0-2) edge[dotted, ->] node[pos=0.55] {\texttt{a}} (dot2)
		(l0-3) edge[dotted, ->] node[pos=0.55] {\texttt{a}} (dot3)
		(l0-4) edge[dotted, ->] node[pos=0.55] {\texttt{a}} (dot4)
		;
		\path[->]
		(l1-1P2-1) edge node[pos=0.25] {\texttt{c}} node[dot,pos=0.5,anchor=center] {} node[pos=0.725] {1} (l3-3)
		(l1-1P2-2) edge node[pos=0.3] {\texttt{b}} node[dot,pos=0.5,anchor=center] {} node[pos=0.725] {1} (l2-4)
		(l1-2P2-1) edge node[pos=0.3] {\texttt{c}} node[dot,pos=0.5,anchor=center] {} node[pos=0.725] {1} (l3-1)
		(l1-2P2-2) edge node[pos=0.3] {\texttt{b}} node[dot,pos=0.5,anchor=center] {} node[pos=0.725] {1} (l2-1)
		(l1-3P2-1) edge node[pos=0.3] {\texttt{b}} node[dot,pos=0.5,anchor=center] {} node[pos=0.725] {1} (l2-2)
		(l1-4P2-1) edge node[pos=0.3] {\texttt{b}} node[dot,pos=0.5,anchor=center] {} node[pos=0.725] {1} (l2-3)
		(l1-4P2-2) edge node[pos=0.3] {\texttt{c}} node[dot,pos=0.5,anchor=center] {} node[pos=0.725] {1} (l3-2)
		(l2-1) edge node {$\tau$} (l0-1)
		(l2-2) edge node {$\tau$} (l0-2)
		(l2-3) edge node {$\tau$} (l0-3)
		(l2-4) edge node {$\tau$} (l0-4)
		(l3-1) edge[loop,out=-70,in=-110,looseness=4.5] node[pos=0.5] {$\soj$} (l3-1)
		(l3-2) edge[loop,out=-70,in=-110,looseness=4.5] node[pos=0.5] {$\soj$} (l3-2)
		(l3-3) edge[loop,out=-70,in=-110,looseness=4.5] node[pos=0.5] {$\soj$} (l3-3)
		;
	\end{tikzpicture}
	\caption{Excerpt of $\da(\mathcal{M}_\mathit{re})$ created with per-interval probability $0.5$. The superscript intervals denote the possible valuations for timers $x$ (on the left) and $y$ (on the right).}
	\label{fig:running_example_after_transformation}
\end{figure}

\begin{example}
\label{ex:RunningAbstraction}
Applying interval abstraction to $\mathcal{M}_\mathit{re}$ of \Cref{ex:RunningSA} with a per-interval probability of $0.5$ results in the MDP of which an excerpt is shown in \Cref{fig:running_example_after_transformation}.
In $\ell_0$, initially, all timers are expired, so the \texttt{a}-labelled edge in the SA can be taken.
It resets both timers, thus the corresponding transition in the MDP as per R1 leads to a distribution over the four combinations of intervals for $x$ and $y$.

Location $\ell_1$ has two edges, with two distinct non-empty guards, $\set{y}$ and $\set{x}$.
By R2, we get up to two $\soj$-labelled transitions out of each state in the second row of \Cref{ex:RunningSA}.
For clarity, we instead label the ones where $\set{y}$ ($\set{x}$) was chosen to expire with $\delta_y$ ($\delta_x$).
Note that in ${\ell_1}^{[3,4] \times [0, 2.5]}_\mathit{R2}$, only the expiration of $\set{y}$ is considered because the expiration of $\set{x}$ would violate the ``$lb(G) \leq ub(G')$'' requirement:
$ub(\set{y}) = 2.5 \leq lb(\set{x}) = 3$, \ie $y$ \emph{must} expire before~$x$.

In the third row, the timers of the chosen guard set have expired, and the corresponding SA edge labelled \texttt{b} or \texttt{c} can be taken as per R1.
Observe that timers not in the chosen guard set may have lower bound $0$ now but are not considered expired, as explained earlier;
thus \eg in state ${\ell_1}^{[0,0.5] \times \expired}_\mathit{R1}$ only the \texttt{b}-labelled transition is possible despite $lb(x) = 0$.

Finally, in the fourth row, we can either loop back to $\ell_0$ from $\ell_2$, or else we are stuck in $\ell_3$.
In the former case, the \texttt{a}-labelled transitions reset both timers, and thus lead into the same distribution (not shown in \Cref{ex:RunningSA}) as the \texttt{a}-labelled transition out of the initial state, looping back into the second row.
\Cref{ex:RunningSA} thus includes all 23 states of the interval abstraction and all transitions.

\end{example}

\subsection{Correctness}
\label{sec:AbstractionCorrectness}

In order to check the optimal probability of reaching a location $\ell$ (or a set of locations) in the SA, we compute the optimal probability of reaching a state $\tuple{\ell, g}$ for any $g$ in the interval abstraction MDP (see \Cref{sec:Background}).
Since the SA is finite, and the abstraction intervals are fixed for every timer, this MDP must also be finite---but it can be large when many different combinations of (post-delay) intervals can be reached.
MDP model checking, however, has good practical scalability, solving models of hundreds of millions of states on standard hardware today.

The interval abstraction does not preserve reachability probabilities exactly, however; it rather introduces an overapproximation from two sources:
First, MDP strategies can choose the precise value to use within the intervals whereas this choice is fully random in the SA semantics.
Second, this choice is actually postponed from the moment the edge that does the sampling is taken to the point that one of the transitions generated by R2 must be chosen.
This moves nondeterministic choices over (other) probabilistic choices, which also gives more information to the strategies.\!\footnote{Consider tossing a coin to \emph{then} nondeterminstically guess its value vs.\ the inverse.}
Thus maximum (minimum) reachability probabilities in the MDP overapproximate (underapproximate) the corresponding probabilities in the SA's semantics:

\begin{theorem}
\label{thm:IntervalApprox}
Given an SA $\mathcal{M}$ as usual and a set of goal locations $\mathcal{G} \subseteq \Locs$, let $p_\mathit{min}^\mathcal{M} = \Pmin{\diamond\,\mathcal{G}}$ and $p_\mathit{max}^\mathcal{M} = \Pmax{\diamond\,\mathcal{G}}$ be the optimal reachability probabilities obtained from $\sem{\mathcal{M}}$ as per \Cref{def:ReachProb}, and let $p_\mathit{min}^{\da(\mathcal{M})}$ and $p_\mathit{max}^{\da(\mathcal{M})}$ be the analogous probabilities obtained from the MDP $\da(\mathcal{M})$. Then
$$
p_\mathit{min}^{\da(\mathcal{M})} \leq p_\mathit{min}^\mathcal{M}
\qquad\text{and}\qquad
p_\mathit{max}^\mathcal{M} \leq p_\mathit{max}^{\da(\mathcal{M})}.
$$
\end{theorem}

\begin{figure}[t]
	\begin{minipage}[b]{0.32\textwidth}
	\raggedright%
	\begin{tikzpicture}[on grid]
		\node[loc] (l0) {$\ell_0\text{-}0{,}0$};
		\node[] (label) [above left=0.5 and 1.95 of l0,anchor=west] {$\sem{\mathcal{M}_\mathit{re}'}$:};
		\coordinate[left=0.3 of l0.west] (start);
		\node[dot] (dot0) [below=0.55 of l0] {};
		\node[loc] (l11) [below left=0.7 and 1.4 of dot0] {$\ell_1\text{-}2{,}1$};
		\node[loc] (l12) [below left=0.7 and 0.0 of dot0] {$\ell_1\text{-}3{,}3$};
		\node[loc] (l13) [below left=0.7 and -1.4 of dot0] {$\ell_1\text{-}4{,}5$};
		\node[loc] (l11d) [below=0.85 of l11] {$\ell_1\text{-}1{,}0$};
		\node[loc] (l12d) [below=0.85 of l12] {$\ell_1\text{-}0{,}0$};
		\node[loc] (l13d) [below=0.85 of l13] {$\ell_1\text{-}0{,}1$};
		\node[] (l11d1) [below=0.75 of l11d] {$\ldots$};
		\node[] (l11d2) [below left=0.75 and 0.5 of l12d] {$\ldots$};
		\node[] (l11d3) [below right=0.75 and 0.5 of l12d] {$\ldots$};
		\node[] (l11d4) [below=0.75 of l13d] {$\ldots$};

		\node[yshift=1.75mm] (dot0arcL) [right=1.25 of dot0,inner sep=0,align=left] {$\textsc{Uni}(2,4){\times}$\\$\textsc{Uni}(1,5)$};

		\node[] (l11dots) [below left=0.7 and 0.7 of dot0] {$\ldots$};
		\node[] (l12dots) [below left=0.7 and -0.7 of dot0] {$\ldots$};
		\node[] (l11tdots1) [below=0.425 of l11dots] {$\ldots$};
		\node[] (l12tdots1) [below=0.425 of l12dots] {$\ldots$};
		\node[] (l11ddots1) [below=0.85 of l11dots] {$\ldots$};
		\node[] (l12ddots1) [below=0.85 of l12dots] {$\ldots$};
		\node[] (l11ddots2) [below left=0.425 and 0.25 of l11ddots1] {$\ldots$};
		\node[] (l12ddots2) [below right=0.425 and 0.25 of l12ddots1] {$\ldots$};
		\node[] (l11tdots2) [below=0.325 of l11ddots2] {$\ldots$};
		\node[] (l12tdots2) [below=0.325 of l12ddots2] {$\ldots$};

		\path[]
		(l0) edge node[pos=0.6,right] {\texttt{a}} (dot0)
		;

		\path[->]
		(start) edge[] (l0)
		(dot0) edge[bend right=20] node[pos=0.3] (dot0arc0) {} (l11)
		(dot0) edge[] (l12)
		(dot0) edge[bend left=20] node[pos=0.3] (dot0arc1) {} (l13)
		(l11) edge[] node[pos=0.45,right] {1} (l11d)
		(l12) edge[] node[pos=0.45,right] {3} (l12d)
		(l13) edge[] node[pos=0.45,right] {4} (l13d)
		(l11d) edge[] node[pos=0.45,left] {\texttt{b}} (l11d1)
		(l12d) edge[] node[pos=0.45,left] {\texttt{b}} (l11d2)
		(l12d) edge[] node[pos=0.45,right] {\texttt{c}} (l11d3)
		(l13d) edge[] node[pos=0.45,right] {\texttt{c}} (l11d4)
		;

		\path[densely dotted]
		(dot0arc0.center) edge[bend right=50] (dot0arc1.center)
		(dot0arc1.center) edge[out=50,in=180] (dot0arcL.west)
		;
	\end{tikzpicture}
	\end{minipage}%
	\begin{minipage}[b]{0.33\textwidth}
	\centering%
	\begin{tikzpicture}[on grid]
		\node[loc] (l0) {$\ell_0\text{-}0{,}0$};
		\node[] (label) [above left=0.2 and 2.05 of l0,anchor=west] {$\da(\sem{\mathcal{M}_\mathit{re}'})$:};
		\coordinate[left=0.3 of l0.west] (start);
		\node[dot] (dot0) [below=0.55 of l0] {};
		\node[odot] (dot01) [below left=0.35 and 0.7 of dot0] {};
		\node[odot] (dot02) [below right=0.35 and 0.7 of dot0] {};
		\node[loc] (l11) [below left=1 and 1.4 of dot0] {$\ell_1\text{-}2{,}1$};
		\node[loc] (l12) [below left=1 and 0.0 of dot0] {$\ell_1\text{-}3{,}3$};
		\node[loc] (l13) [below left=1 and -1.4 of dot0] {$\ell_1\text{-}4{,}5$};
		\node[loc] (l11d) [below=0.85 of l11] {$\ell_1\text{-}1{,}0$};
		\node[loc] (l12d) [below=0.85 of l12] {$\ell_1\text{-}0{,}0$};
		\node[loc] (l13d) [below=0.85 of l13] {$\ell_1\text{-}0{,}1$};
		\node[] (l11d1) [below=0.75 of l11d] {$\ldots$};
		\node[] (l11d2) [below left=0.75 and 0.5 of l12d] {$\ldots$};
		\node[] (l11d3) [below right=0.75 and 0.5 of l12d] {$\ldots$};
		\node[] (l11d4) [below=0.75 of l13d] {$\ldots$};

		\node[yshift=1mm] (dot01arcL) [left=0.875 of dot01,inner sep=0,align=right] {$[2{,}4]$\\${\times}[1{,}3]$};
		\node[yshift=1mm] (dot02arcL) [right=0.875 of dot02,inner sep=0,align=left] {$[2{,}4]{\times}$\\$[3{,}5]$};

		\node[] (l11dots) [below=0.65 of dot01] {$\ldots$};
		\node[] (l12dots) [below=0.65 of dot02] {$\ldots$};
		\node[] (l11tdots1) [below=0.425 of l11dots] {$\ldots$};
		\node[] (l12tdots1) [below=0.425 of l12dots] {$\ldots$};
		\node[] (l11ddots1) [below=0.85 of l11dots] {$\ldots$};
		\node[] (l12ddots1) [below=0.85 of l12dots] {$\ldots$};
		\node[] (l11ddots2) [below left=0.425 and 0.25 of l11ddots1] {$\ldots$};
		\node[] (l12ddots2) [below right=0.425 and 0.25 of l12ddots1] {$\ldots$};
		\node[] (l11tdots2) [below=0.325 of l11ddots2] {$\ldots$};
		\node[] (l12tdots2) [below=0.325 of l12ddots2] {$\ldots$};

		\path[]
		(l0) edge node[pos=0.6,right] {\texttt{a}} (dot0)
		(dot0) edge[bend right=20] node[pos=0.15,left] {$0.5~~$} (dot01)
		(dot0) edge[bend left=20] node[pos=0.15,right] {$~~0.5$} (dot02)
		;

		\path[->]
		(start) edge[] (l0)
		(dot01) edge[bend right=20] node[pos=0.5] (dot01arc0) {} (l11)
		(dot01) edge[bend left=20] node[pos=0.5] (dot01arc1) {} (l12)
		(dot02) edge[bend right=20] node[pos=0.5] (dot02arc0) {} (l12)
		(dot02) edge[bend left=20] node[pos=0.5] (dot02arc1) {} (l13)
		(l11) edge[] node[pos=0.45,right] {1} (l11d)
		(l12) edge[] node[pos=0.45,right] {3} (l12d)
		(l13) edge[] node[pos=0.45,right] {4} (l13d)
		(l11d) edge[] node[pos=0.45,left] {\texttt{b}} (l11d1)
		(l12d) edge[] node[pos=0.45,left] {\texttt{b}} (l11d2)
		(l12d) edge[] node[pos=0.45,right] {\texttt{c}} (l11d3)
		(l13d) edge[] node[pos=0.45,right] {\texttt{c}} (l11d4)
		;

		\path[densely dotted]
		(dot01arc0.center) edge[bend right=50] (dot01arc1.center)
		(dot01arc0.center) edge[out=135,in=0] (dot01arcL.east)
		(dot02arc0.center) edge[bend right=50] (dot02arc1.center)
		(dot02arc1.center) edge[out=45,in=180] (dot02arcL.west)
		;
	\end{tikzpicture}
	\end{minipage}%
	\begin{minipage}[b]{0.35\textwidth}
	\raggedleft%
	\begin{tikzpicture}[on grid]
		\node[loc] (l0) {$\ell_0\text{-}\expired {\times} \expired$};
		\node[] (label) [above left=0.3 and 2.15 of l0,anchor=west] {$\da(\mathcal{M}_\mathit{re}')$:};
		\coordinate[left=0.3 of l0.west] (start);
		\node[dot] (dot0) [below=0.55 of l0] {};
		\node[loc] (l11) [below left=0.5 and 1.1 of dot0] {$\ell_0\text{-}[2{,}4]{\times}[1{,}3]$};
		\node[loc] (l12) [below right=0.5 and 1.1 of dot0] {$\ell_0\text{-}[2{,}4]{\times}[3{,}5]$};
		\node[loc] (l112) [below right=0.8 and 0.3 of l11.south east,anchor=north east] {$\ell_0\text{-}\expired{\times}[0{,}1]$};
		\node[loc] (l111) [below=0.4 of l11.south west,anchor=north west] {$\ell_0\text{-}[0{,}3]{\times}\expired$};
		\node[loc] (l122) [below=0.8 of l12.south east,anchor=north east] {$\ell_0\text{-}\expired{\times}[0{,}3]$};
		\node[loc] (l121) [below left=0.4 and 0.3 of l12.south west,anchor=north west] {$\ell_0\text{-}[0{,}1]{\times}\expired$};

		\node[] (l111dots) [below left=2.0 and 0.5 of l11] {$\ldots$};
		\node[] (l112dots) [below right=2.0 and 0.5 of l11] {$\ldots$};
		\node[] (l121dots) [below left=2.0 and 0.5 of l12] {$\ldots$};
		\node[] (l122dots) [below right=2.0 and 0.5 of l12] {$\ldots$};

		\path[]
		(l0) edge node[pos=0.6,right] {\texttt{a}} (dot0)
		;

		\path[->]
		(start) edge[] (l0)
		(dot0) edge[bend right=15] node[pos=0.15,left,inner sep=3mm] {$0.5~~$} (l11)
		(dot0) edge[bend left=15] node[pos=0.15,right,inner sep=3mm] {$~~0.5$} (l12)
		(l11) edge[bend right=15] node[left] {$\sigma$} (l111)
		(l11) edge[bend left=25,transform canvas={xshift=2.75mm}] node[left,pos=0.26] {$\sigma$} (l112)
		(l12) edge[bend right=15] node[left] {$\sigma$} (l121)
		(l12) edge[bend left=15,transform canvas={xshift=2.5mm}] node[left,pos=0.225] {$\sigma$} (l122)
		(l111) edge[bend right=30] node[pos=0.725,left] {\texttt{b}} (l111dots)
		(l112) edge[bend right=5] node[pos=0.45,left] {\texttt{c}} (l112dots)
		(l121) edge[bend right=35] node[pos=0.725,left] {\texttt{b}} (l121dots)
		(l122) edge[bend right=8] node[pos=0.45,left] {\texttt{c}} (l122dots)
		;
	\end{tikzpicture}
	\end{minipage}
	\caption{Transformations of the semantics of the SA $\mathcal{M}_\mathit{re}'$.}
	\label{fig:ProofIllustrations}
\end{figure}
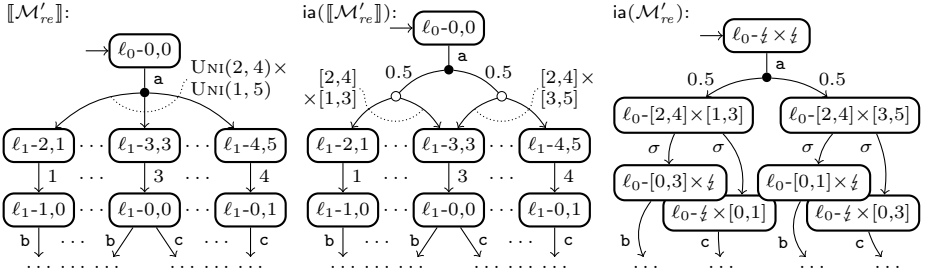

\begin{proof}[sketch]
Our proof is inspired by that of \cite[Thm.~4.22]{Hah13}.
We illustrate the idea with a variant $\mathcal{M}_\mathit{re}'$ of $\mathcal{M}_\mathit{re}$ from \Cref{fig:sa_running_example} with CPDs slightly simplified to $x\colon \textsc{Uni}(2, 4)$ and $y\colon \textsc{Uni}(1, 3)$, and abstraction intervals $\Part_\mathit{re}'(x) = \set{ [2,4] }$ and $\Part_\mathit{re}'(y) = \set{ [1,3], [3, 5] }$.
We show an excerpt of $\sem{\mathcal{M}_\mathit{re}'}$ on the left of \Cref{fig:ProofIllustrations}.

The interval abstraction directly maps an SA $\mathcal{M}$ to the MDP $\da(\mathcal{M})$. %
We introduce an intermediate step:
an interval abstraction $\da(\sem{\mathcal{M}})$ of $\sem{\mathcal{M}}$ over the same $\Part$. %
It is obtained from $\sem{\mathcal{M}}$ by replacing each continuous distribution $\mu_{x_i}$ in rule $\mathrm{S1}$ of \Cref{def:SASemantics} by a discrete distribution $\mu_{x_i}^\da$ over the intervals in $\Part(x_i)$ followed by a continuous nondeterministic selection of a concrete value from the chosen interval such that $\mu_{x_i}^\da(I) = \mu_{x_i}(I)$.
We show $\da(\sem{\mathcal{M}}_\mathit{re}')$ in the middle of \Cref{fig:ProofIllustrations}.
Where a transition in $\sem{\mathcal{M}}$ leads into a continuous distribution over states, in $\da(\sem{\mathcal{M}})$, it now leads into a discrete distribution (solid dot \raisebox{0.67pt}{$\tikz{\node[dot]{}}$}) over sets of states out of which one is chosen nondeterministically (outlined dots \raisebox{0.67pt}{$\tikz{\node[odot]{}}$}).

The (reachable) states of $\sem{\mathcal{M}}$ and $\da(\sem{\mathcal{M}})$ are the same, as is the transition graph structure.
Given a strategy $\sched$ for $\sem{\mathcal{M}}$, we can construct a randomised strategy $\sched_{\sem{\da}}$ for $\da(\sem{\mathcal{M}})$ that induces the exact same probability measure:
First, fix $\sched(s) = \sched_{\sem{\da}}(s)$ for all states~$s$.
Then, $\sched_{\sem{\da}}$ resolves all the new continuous nondeterministic choices of values from intervals randomly so that $\sched_{\sem{\da}} \circ \mu_{x_i}^\da  = \mu_{x_i}$ everywhere.
This is possible because $\mu_{x_i}^\da$ and $\mu_{x_i}$ assign the same probabilities to all these intervals.
$\da(\sem{\mathcal{M}})$ may have additional strategies not corresponding to any in $\sem{\mathcal{M}}$, and we thus have
$p_\mathit{min}^{\da(\sem{\mathcal{M}})} \leq p_\mathit{min}^\mathcal{M}$
and
$p_\mathit{max}^\mathcal{M} \leq p_\mathit{max}^{\da(\sem{\mathcal{M}})}$.

Observe that, by rule $\mathrm{S2}$ of \Cref{def:SASemantics}, $\da(\sem{\mathcal{M}})$ contains no nondeterminism over delays, delay transitions are deterministic, and only connect states that are in the same SA location.
We can thus collapse all delay transitions, \ie delete them and copy the outgoing transitions of their successor states onto their source states, without affecting the trace distributions modulo delay steps of $\da(\sem{\mathcal{M}})$.
Denote the result as $\da(\sem{\mathcal{M}})'$.
Now, $\da(\mathcal{M})$ weakly probabilistically simulates $\da(\sem{\mathcal{M}})'$ with the $\delta$ transitions being unobservable steps (that subsume the deterministic delays of $\da(\sem{\mathcal{M}})$ but also pull forward the subsequent nondeterministic choice of action-labelled transition).
This gives us
$$
p_\mathit{min}^{\da(\mathcal{M})} \leq p_\mathit{min}^{\da(\sem{\mathcal{M}})'} = p_\mathit{min}^{\da(\sem{\mathcal{M}})}
\qquad\text{and}\qquad
p_\mathit{max}^{\da(\sem{\mathcal{M}})} = p_\mathit{max}^{\da(\sem{\mathcal{M}})'} \leq p_\mathit{max}^{\da(\mathcal{M})}.
\eqno\qed
$$
\end{proof}

\subsection{Refinement}

As we make the abstraction intervals smaller, the strategy in $\da(\sem{\mathcal{M}})$ gets fewer new choices, and the overlap between intervals in $\da(M)$ reduces.
We thus expect that, the smaller the intervals are, the more closely we approximate true optimal probabilities.
This would give rise to a refinement-based approach where we start with coarse intervals, and repeatedly split them until \eg the max.\ probability to reach an unsafe state can be shown to be below a specified safety threshold.

However, although we conjecture that the abstraction error goes to zero as the width of the intervals goes to zero, it is not yet clear how to prove this.
In particular, the effect of reducing interval width on the error is not monotonic:
Consider another variant of $\mathcal{M}_\mathit{re}$ where $x\colon \textsc{Uni}(3, 4)$ and $y\colon \textsc{Uni}(0, 4)$, the $\tau$-labelled edge is removed, and we are interested in $\Pmax{\diamond\,\set{\ell_3}}$ (\ie $x$ wins the race).
We have $\Pmax{\diamond\,\set{\ell_3}} = \frac{1}{8}$.
With timer-domain partition $\Part_1$ given by $\Part_1(x) = \set{[3,4]}$ and $\Part_1(y) = \set{[0,2], [2,4]}$, we get max.\ probability $\frac{1}{2}$:
Iff $[2,4]$ is sampled for $y$, which happens with probability $\frac{1}{2}$, we have a choice to let $x$ expire first, because it is the only interval for $y$ that overlaps with the one interval $[3,4]$ for $x$.
With $\Part_2$ given by $\Part_2(x) = \set{[3,4]}$ and $\Part_2(y) = \set{[0,1], [1, 3], [3,4]}$, which has smaller-or-equal intervals than $\Part_1$, we however get \emph{higher} (\ie worse) max.\ probability $\frac{3}{4}$ as $x$ can be chosen to win when either $[1, 3]$ or $[3,4]$ is sampled.

We suspect that monotonicity can be achieved with the right notion of ``refinement'', such as never removing interval bounds as in our example, which could open the way for a convergence proof.
Generating intervals based on per-interval probabilities, however, does not guarantee such a kind of refinement.

\section{Modelling and Model Exchange}
\label{sec:Languages}

Specifying SA mathematically as in \Cref{def:sa} or via drawings as in \Cref{fig:sa_running_example} is not convenient to build large SA models of real-life systems:
We would want to use parallel composition, and have the ability to use discrete (\eg Boolean and integer) variables to express data-dependent control flow.
The only modelling language for SA providing these facilities so far was IOSA as used by \tool{Fig}, which is an extension of the \tool{Prism} language~\cite{KNP11}.
It is specific to the input-output and weakly deterministic nature of the SA that \tool{Fig} requires.

\begin{wrapstuff}[l,type=figure,width=5.25cm,top=7]
\vspace{-0.5em}
\begin{lstlisting}[language=Modest,escapechar=`,morekeywords={},numbers=none,xleftmargin=0pt]
option "sa";
action a, b, c;
timer x, y;

do {
   a {= x = Uni(2, 4),
        y = Uni(0, 5) =};
   alt {
   :: when(y <= 0) b; tau
   :: when(x <= 0) c; stop
   }
}
\end{lstlisting}
\vspace{-1.25em}
\caption{\modest model for $\mathcal{M}_\mathit{re}$.}
\label{fig:RunningExampleModest}
\end{wrapstuff}

As part of our work on implementing model checking for general SA via interval abstraction, we also extended the \modest language to support SA directly.
\modest is a high-level modelling language inspired by process algebra but with a syntax akin to common programming languages; it provides standard control flow constructs as well as exception handling, process declarations and calls to reuse code, and user-defined datatypes and functions.
We mainly added a new \lstinline{timer} type, and an \lstinline{"sa"} option to request the maximal progress semantics of SA.
Guard sets and resets can then be expressed using the standard \modest guard expressions and assignment constructs.
This makes modelling SA in \modest slightly more flexible than our \Cref{def:sa} as different distributions can be used for the same timer in different assignments; this can be mapped back to \Cref{def:sa} by introducing additional timers if necessary.
In \Cref{fig:RunningExampleModest}, we show one way to express SA $\mathcal{M}_\mathit{re}$ of \Cref{ex:RunningSA} in \modest.

Implementing a new tool for the analysis of SA would now require implementing a parser for either IOSA (if its restrictions are acceptable) or \modest, which is a large and tedious task in either case.
To reduce such burden on tool developers, the JSON-based \jani format was designed as a model interchange format between tools.
We also extended \jani with support for timers and a new \texttt{sa} model type, allowing \modest SA models to be translated to \jani, which could then be read by other future SA analysis tools.

\section{Experiments}
\label{sec:Experiments}

To evaluate the scalability and performance of our new SA model checking approach, we created a prototype implementation and applied it to compute maximal and minimal reachability probabilities for three sets of example SA:
Our running example and the small SA that distinguish strategy classes in~\cite{DGHS18},
a set of (Markovian and non-Markovian) queueing models, and variants of
the file server model from~\cite{HHH14}.
These examples cover all common features of SA that we support:
discrete nondeterministic choices, discrete probability distributions over target locations, immediate and timer-guarded edges, including guards with multiple timers and races between edges, and timers following different CPDs.

\subsection{Implementation}
\label{sec:Implementation}

We implemented a prototype tool for our SA analysis approach in Rust.
We chose Rust as a future-proof language for a new tool that delivers both high performance and a reasonable safety level, but note that our prototype is not yet highly optimised.
It takes as input a stochastic automaton specified in a simple textual file format and 
for each timer~$x$ the per-interval probability $p_I^x$ that determines how its distribution's support is abstracted into intervals.
We extended the \toolset with a conversion from \modest SA models into this format.

Using the given $p_I^x$ for timer~$x$ with cumulative distribution function (CDF) $F_x$, %
the tool computes the intervals for each distribution by partitioning the support into $ 1/p_I^x $ intervals of equal probability mass $p_I^x$, with boundaries at the quantiles $F_x^{-1}(k \cdot p_I^x)$ for $k = 1, 2, \dots$. If $1/p_I^x$ is not an integer, the last interval covers the remainder of the distribution's support extending to the upper end of the support, which may be $\infty$ for distributions such as exponential or Weibull, and might have a probability smaller than $p_I^x$.
The tool currently supports timers adhering to uniform, exponential, Erlang, and Weibull distributions.
Adding support for a new distribution is straightforward if there exists a way to compute its inverse cumulative distribution function $F_x^{-1}$, which is necessary to find intervals of probability mass $p_I^x$ that partition the distribution's support.
For the Erlang distribution, we currently use the \texttt{ruststat} library's implementation of the inverse CDF for the gamma distribution.

In addition to reachability probabilities, the tool also supports computing the probabilities of paths satisfying a $\neg\mathit{avoid} \mathbin{\mathsf{U}} \mathcal{G}$ until formula. %
After performing interval abstraction as described in \Cref{sec:OurIntervalAbstraction}, the tool uses standard VI on the resulting MDP to compute the minimum and maximum reachability probabilities.
It uses an absolute difference of $\epsilon = 10^{-6}$ stopping criterion for VI.
We plan to replace VI by the sound and floating point-safe version of interval iteration implemented in \tool{mcsta}~\cite{Har22} as the tool matures.
In addition to returning the computed probabilities, the tool can also export the interval abstraction MDP in Graphviz\ \texttt{.dot} format.

\subsection{Experimental Setup}
\label{sec:ExperimentalSetup}

Our experiments were performed on an AMD Ryzen 7 PRO 7730U (2.0-4.5\,GHz) laptop with 32\,GB of RAM running openSuse Leap 16 Linux.
The tool was built with Rust's default release build profile.
We measured the total tool runtime with the Linux \texttt{time} command; the tool additionally outputs the time it took for parsing, computing the intervals via CDF quantiles, constructing the interval abstraction MDP from the SA, and running VI for both minimum and maximum probability.
The first two steps always took negligible time in our experiments.
In our tables, we round all probabilities to three decimal places.

\subsection{Experimental Results}
\label{sec:ExperimentalResults}

\subsubsection{Small SA.}

We first ran the tool on $\mathcal{M}_\mathit{re}$ from \Cref{ex:RunningSA} and SA $M_1$, $M_2$, and $M_5$ from~\cite{DGHS18} (shown in \Cref{app:FossacsAutomata}).
Of the SA from~\cite{DGHS18}, we chose the three where the optimal reachability probabilities of interest are nontrivial, \ie not $0$ and $1$ for minimum and maximum, respectively, for the strategy class that our approach effectively uses (memoryless prophetic with full information).

\begin{table}[t]
	\centering
\caption{Results for the running example and the SA from~\cite{DGHS18}.}
\label{tab:results_small_automata}
		\scriptsize
	  \setlength{\tabcolsep}{6pt}
		\renewcommand*{\arraystretch}{1.1}
		\begin{tabular}{l*{4}{cc}}
			\toprule
			Model & \multicolumn{2}{c}{$\mathcal{M}_\mathit{re}$} & \multicolumn{2}{c}{$M_1$} & \multicolumn{2}{c}{$M_2$} & \multicolumn{2}{c}{$M_5$} \\
			\cmidrule(lr){2-3}
			\cmidrule(lr){4-5}
			\cmidrule(lr){6-7}
			\cmidrule(lr){8-9}
			Property   & \pmin   & \pmax   & \pmin   & \pmax   & \pmin   & \pmax  & \pmin   & \pmax  \\
			\midrule
			True probability  &  1\phantom{.000} & 1\phantom{.000} & 0.25\phantom{0} & 0.75\phantom{0} & 0.03125 & 0.96875 & 0.125  & 0.875 \\
			\midrule
			  Result  & 1.000  & 1.000  &  0.245 & 0.755 & 0.004\phantom{00}  & 0.996\phantom{00}  & 0.117  & 0.883  \\
			$p_I$  & \multicolumn{2}{c}{0.05}  & \multicolumn{2}{c}{0.01} & \multicolumn{2}{c}{0.05} & \multicolumn{2}{c}{0.0125}  \\
			MDP states  & \multicolumn{2}{c}{2378} & \multicolumn{2}{c}{82657} & \multicolumn{2}{c}{467159} & \multicolumn{2}{c}{260395} \\
			Total tool runtime  & \multicolumn{2}{c}{0.2\,s} & \multicolumn{2}{c}{1.5\,s} & \multicolumn{2}{c}{4.2\,s} & \multicolumn{2}{c}{20.6\,s}  \\

			\bottomrule
		\end{tabular}
\end{table}

The results are shown in \Cref{tab:results_small_automata}.
For $\mathcal{M}_\mathit{re}$, we use $\mathcal{G} = \set{\ell_3}$.
It is easy to see that the minimum and maximum probabilities to reach $\ell_3$ are both $1$ (as there is a positive probability to go to $\ell_3$ in each loop iteration, and we consider unbounded reachability), and
our tool correctly approximates them (up to a small VI error). %
\begin{table}[t]
	\centering
	\caption{Results for $M_1$ from~\cite{DGHS18} with different per-interval probabilities.}
	\label{tab:results_m1}
		\scriptsize
	  \setlength{\tabcolsep}{2.5pt}
		\renewcommand*{\arraystretch}{1.1}
		\begin{tabular}{l*{5}{cc}}
			\toprule

			$p_I$ & \multicolumn{2}{c}{0.5} & \multicolumn{2}{c}{0.1} & \multicolumn{2}{c}{0.05} & \multicolumn{2}{c}{0.01} & \multicolumn{2}{c}{0.005} \\
			\cmidrule(lr){2-3}
			\cmidrule(lr){4-5}
			\cmidrule(lr){6-7}
			\cmidrule(lr){8-9}
			\cmidrule(lr){10-11}
			Property   & \pmin   & \pmax   & \pmin   & \pmax   & \pmin   & \pmax  & \pmin   & \pmax & \pmin   & \pmax   \\
			\midrule

			Result  & 0.000  & 1.000  &  0.200 & 0.800 & 0.225  & 0.775  & 0.245  & 0.755 & 0.247 & 0.752\\
			MDP states  & \multicolumn{2}{c}{45} & \multicolumn{2}{c}{957} & \multicolumn{2}{c}{3589} & \multicolumn{2}{c}{82657} & \multicolumn{2}{c}{325565} \\
			Total tool runtime  & \multicolumn{2}{c}{0.08\,s} & \multicolumn{2}{c}{0.07\,s} & \multicolumn{2}{c}{0.09\,s} & \multicolumn{2}{c}{1.6\,s} & \multicolumn{2}{c}{13.1\,s} \\

			\bottomrule
		\end{tabular}
\end{table}
We also correctly under- and overapproximate the true probabilities for $M_1$, $M_2$, and $M_5$ (see \Cref{app:FossacsAutomata} for the analytical calculations).
These three automata are constructed in such a way that, in states of the abstraction MDP where the intervals for the model's timers overlap (\eg $b(x) = [0, 0.5]$ and $b(y) = [0, 1]$), a prophetic strategy can force the goal to be reached (or not).
Thus, the smaller we make per-interval probability, the more intervals we have at our disposal, and the lower is the probability (in the MDP) that the timers are reset to overlapping intervals, so we can more precisely approximate the actual probabilities. %
We investigate this effect for $M_1$ by varying $p_I$, with the results shown in \Cref{tab:results_m1}.
As expected, we see that, as $p_I$ decreases, $\pmax$ decreases and $\pmin$ increases.
Overall, on the models from~\cite{DGHS18}, we obtain good approximations of the actual results in short runtimes despite the overapproximation from the abstraction. %

\subsubsection{Queues.}

Queuing systems, in particular non-Markovian ones, can be modelled by SA and are interesting because every location features a race condition between two timers for arrival and service.
We consider two systems with bounded buffers so we get finite SA:
An $M/M/1/100$ queue with $\mathit{Exp}(2)$-distributed interarrival and $\mathit{Exp}(0.5)$-distributed service times
and a $\mathit{Weibull}/E_6/1/100$ queue with $\mathit{Weibull}(2, 1.5)$ interarrival and $\mathit{Erlang}(6, 2.5)$ service times.
The two SA's locations encode the number of jobs in the queue and the number of jobs completed so far.
We compute the probability of reaching a queue overflow before a given number of jobs is completed.
As there is no nondeterminism in these models, our bounds on $\pmin$ and $\pmax$ under- and overapproximate the one true probability, respectively, and we thus report our result as intervals.

\begin{table}[t]
	\centering
	\caption{Results for $M/M/1/100$ and $\mathit{Weibull}/E_6/1/100$ queues with $p_I = 0.1$.}
	\label{tab:results_queues}
	\scriptsize
	\setlength{\tabcolsep}{6pt}
	\renewcommand*{\arraystretch}{1.1}
	\begin{tabular}{l*{4}{cc}}
		\toprule
		\multicolumn{1}{c}{} & \multicolumn{4}{c}{M/M/1/100}& \multicolumn{4}{c}{Weibull/$E_6$/1/100} \\
		\cmidrule(lr){2-5}
		\cmidrule(lr){6-9}
		Job bound & \multicolumn{2}{c}{150} & \multicolumn{2}{c}{200} & \multicolumn{2}{c}{150} & \multicolumn{2}{c}{200} \\
		\midrule

		Result  & \multicolumn{2}{c}{$[0.000, 0.979]$}  &  \multicolumn{2}{c}{$[0.000, 0.263]$} & \multicolumn{2}{c}{$[0.000, 0.924]$}  & \multicolumn{2}{c}{$[0.000, 0.084]$}  \\
		SA locations  & \multicolumn{2}{c}{10202}  & \multicolumn{2}{c}{15252} & \multicolumn{2}{c}{10202} & \multicolumn{2}{c}{15252}  \\
		MDP states  & \multicolumn{2}{c}{2229270} & \multicolumn{2}{c}{3334270} & \multicolumn{2}{c}{2229270} & \multicolumn{2}{c}{3334270} \\
        Interval abstraction  & \multicolumn{2}{c}{\phantom{0}46.3\,s} & \multicolumn{2}{c}{\phantom{0}67.3\,s} & \multicolumn{2}{c}{\phantom{0}47.2\,s} & \multicolumn{2}{c}{\phantom{0}64.9\,s}  \\
        Value iteration  & \multicolumn{2}{c}{\phantom{0}44.7\,s} & \multicolumn{2}{c}{\phantom{0}70.0\,s} & \multicolumn{2}{c}{\phantom{0}46.0\,s} & \multicolumn{2}{c}{\phantom{0}66.7\,s}  \\
		Total tool runtime  & \multicolumn{2}{c}{\phantom{0}95.0\,s} & \multicolumn{2}{c}{143.8\,s} & \multicolumn{2}{c}{\phantom{0}96.8\,s} & \multicolumn{2}{c}{137.4\,s}  \\

		\bottomrule
	\end{tabular}
\end{table}

The results of applying our tool to these SA are shown in \Cref{tab:results_queues}.
Both models have a buffer size of 100 and were analysed with a bound of 150 and 200 arrivals each.
While for 150 arrivals the maximal probability for finishing all jobs without an overflow is relatively high for both models, %
it decreases rapidly for 200 arrivals. %
The table also shows the number of locations of the SA and the number of states of the interval abstraction MDP. %
What makes the MDPs rather large despite using only 10 intervals per timer is the fact that almost every location has two outgoing edges guarded by one timer---one for the next arrival and one for the service of the current job---and timers are reset on every edge.
Despite having to generate and analyse MDPs with millions of states, the runtime from the prototype tool on a standard laptop remains below three minutes for the transformation and VI combined.

\subsubsection{File server.}
\label{SubSec:FileServer}

\begin{table}[t]
	\centering
	\caption{Results for the file server models compared to \tool{mcsta}'s STA-based approach.}
	\label{tab:results_fileserver}
	\scriptsize
	\setlength{\tabcolsep}{3.5pt}
	\renewcommand*{\arraystretch}{1.2}
		\begin{tabular}{llccrrcrrc}
			\toprule
			&& \multicolumn{4}{c}{SA interval abstraction (this paper)} & \multicolumn{3}{c}{\tool{mcsta}~\cite{HHH14}} & \tool{modes} \\
			\cmidrule(lr){3-6}\cmidrule(lr){7-9}\cmidrule(lr){10-10}
			var. & $C\text{-}n$ & intervals & result & states & time & result & states & time & result \\
			\midrule
			unif. & $3\text{-}5$ & $\frac{1}{\,8\,}$/$\frac{1}{10}$/$\frac{1}{2}$ & $[0.316, 0.359]$ & 32690
 & 0.2\,s & && & \multirow{3}{*}{$0.340$} \\
			      &               & $\frac{1}{16}$/$\frac{1}{10}$/$\frac{1}{2}$ &  $[0.327, 0.350]$ & 45460
 & 0.4\,s & $[0.320, 0.362]$ & 362389 & 3.0\,s\\
			      &               & $\frac{1}{32}$/$\frac{1}{10}$/$\frac{1}{2}$ & $[0.331, 0.348]$ & 228696
 & 2.3\,s\\[2pt]
			      & $5\text{-}8$  &  $\frac{1}{\,8\,}$/$\frac{1}{10}$/$\frac{1}{2}$ & $[0.177, 0.338]$ & 158942
 & 1.3\,s & && & \multirow{3}{*}{$0.255$}\\
                   &               & $\frac{1}{16}$/$\frac{1}{10}$/$\frac{1}{2}$ & $[0.210, 0.301]$ & 229211 & 2.4\,s & $[0.197, 0.315]$ & 1209231 & 8.9\,s\\
			       &               & $\frac{1}{32}$/$\frac{1}{10}$/$\frac{1}{2}$ & $[0.224, 0.286]$ & 1252056
 & 17.4\,s\\[2pt]
                  & $5\text{-}10$  &  $\frac{1}{\,8\,}$/$\frac{1}{10}$/$\frac{1}{2}$ & $[0.307, 0.512]$ & 255844
 & 2.3\,s & && & \multirow{3}{*}{$0.411$}\\
                   &               & $\frac{1}{16}$/$\frac{1}{10}$/$\frac{1}{2}$ & $[0.351, 0.469]$ & 396175
 & 4.2\,s & $[0.336, 0.485]$ & 1633339 & 13.6\,s\\
			      &               & $\frac{1}{32}$/$\frac{1}{10}$/$\frac{1}{2}$ & $[0.370, 0.451]$ & 2314835
 & 38.0\,s\\[2pt]
			\midrule
            exp. & $3\text{-}5$  & $\frac{1}{\,8\,}$/$\frac{1}{2}$/$\frac{1}{2}$ & $[0.230, 0.373]$ & 724557 & 5.3\,s & \multirow{3}{*}{$[0.261, 0.371]$} & \multirow{3}{*}{1016802} & \multirow{3}{*}{2.8\,s} & \multirow{3}{*}{$0.311$} \\
			      &               & $\frac{1}{10}$/$\frac{1}{4}$/$\frac{1}{2}$ & $[0.244, 0.369]$ & 2130622
 & 18.1\,s \\
			      &               & $\frac{1}{16}$/$\frac{1}{4}$/$\frac{1}{2}$ & $[0.267, 0.354]$ & 19602981
 & 175.5\,s\\[2pt]
			      & $5\text{-}8$  & $\frac{1}{4}$/$\frac{1}{4}$/$\frac{1}{2}$ & $[0.088, 0.472]$ & 1624440 & 9.4\,s & \multirow{2}{*}{$[0.206, 0.339]$} & \multirow{2}{*}{3064195} & \multirow{2}{*}{8.2\,s} & \multirow{2}{*}{$0.284$}\\
                   &               & $\frac{1}{5}$/$\frac{1}{2}$/$\frac{1}{2}$ & $[0.108, 0.459]$ & 2640688 & 16.3\,s\\[2pt]
                  & $5\text{-}10$  & $\frac{1}{4}$/$\frac{1}{4}$/$\frac{1}{2}$ & $[0.139, 0.659]$ & 7612384 & 57.0\,s & \multirow{2}{*}{$[0.304, 0.508]$} & \multirow{2}{*}{4139557} & \multirow{2}{*}{12.0\,s} & \multirow{2}{*}{$0.437$}\\
                   &               & $\frac{1}{5}$/$\frac{1}{2}$/$\frac{1}{2}$ & $[0.171, 0.644]$ & 14334125 & 110.6\,s\\
		\end{tabular}
\end{table}

Finally, we check the model of a single-threaded file server with slow archival storage and queue length~$C$ introduced in~\cite{HHH14}.
The service times (timer $x_\mathit{out}$) depend on the file sizes and are assumed to be uniformly distributed over $[1,3]$.
A requested file is either on disk or in an archive with probability $0.8$ and $0.2$, respectively. 
Retrieving a file from the archive takes between 30 and 40 time units, which is continuously nondeterministic in the original model; we use a uniform distribution over $[30, 40]$ (timer $x_\mathit{a}$) instead.
As a result, there is no more nondeterminism in the model.
We additionally replace the discretely uniformly distributed initial queue length of the original model by starting with an empty queue.
We ask for the probability of processing a given number $n$ of arrivals without having a queue overflow; due to the absence of nondeterminism, our bounds on min.\ and max.\ probability again bound this one value.

Based on the original \modest STA model file from~\cite{HHH14}, we created four \modest models representing two variants, one where the interarrival times (timer $x_\mathit{in}$) are $\mathit{Exp}(\frac{1}{8})$-distributed and one where they are $\mathit{Uni}(0,16)$-distributed, each as an STA and as an SA model.
We first compute reference results using the \tool{modes} statistical model checker~\cite{BDHS20} with 10 million simulation runs for high accuracy.
We can then compare the results and performance of our approach for SA with those of \tool{mcsta}'s digital clocks-based approach for STA.
Our tool is parametrised by a per-interval probability $p_I^x$ for each timer $x$, which results in $\lceil 1/p_I^x \rceil$ intervals of possibly different widths but equal probability mass (except for the final one).
\tool{mcsta} on the other hand takes a residual probability $p_r$ and generates as many intervals of width $1$ (as it uses digital clocks) as needed for the last (possibly open) interval to have probability mass $\leq p_r$.
For the $\mathit{Uni}(0,16)$ variant, using $p_I^{x_\mathit{in}} = \frac{1}{16}$, $p_I^{x_\mathit{a}} = \frac{1}{10}$, and $p_I^{x_\mathit{out}} = \frac{1}{2}$ with our tool (abbreviated $\frac{1}{16}$/$\frac{1}{10}$/$\frac{1}{2}$) results in the same intervals as \tool{mcsta}.
For the $\mathit{Exp}(\frac{1}{8})$ variant, we cannot make the intervals match given the current tool parametrisations;
here, \tool{mcsta} generates 25 intervals for $x_\mathit{in}$ with its default $p_r = 0.05$, so we would need $p_I^{x_\mathit{in}} = \frac{1}{25}$ to generate the same \emph{number} of intervals, which would however have very different bounds and cover much more of the support of $\mathit{Exp}(\frac{1}{8})$.

Our results are shown in \Cref{tab:results_fileserver}, for different values of $C$ and $n$.
All results of our tool and \tool{mcsta} correctly bound the reference values.
In our tool, runtime was again about equally divided between interval abstraction and VI.
For the $\mathit{Uni}(0,16)$ variant, we consistently obtain more accurate results in less than $30\,\%$ of the runtime compared to \tool{mcsta}.
The MDPs we check are much smaller for the same intervals because we use big time steps while \tool{mcsta}'s digital clocks approach has time steps of duration~1.
In contrast to \tool{mcsta}, we can refine the intervals per timer.
We found that refining intervals pays off most for $x_\mathit{in}$, significantly improving accuracy over \tool{mcsta}.

Results for the $\mathit{Exp}(\frac{1}{8})$ variant show that the choice of intervals matters, and in this case, \tool{mcsta}'s width-1 intervals concentrated on lower delay values for $x_\mathit{in}$ work better:
To achieve similar accuracy, we need over $6\times$ the runtime.
Our interval abstraction MDPs are much larger despite using fewer intervals (as we do not get close to $p_I^{x_\mathit{in}} = \frac{1}{25}$ without running out memory).
We suspect that generating less regular equal-probability intervals for $\mathit{Exp}(\frac{1}{8})$ leads to some very wide intervals that cause others to significantly expand in delay steps (reducing accuracy), as well as causing more combinations of differently-sized intervals in the state space instead of going back to the same bounds often (blowing up the MDP).
We thus plan to implement \tool{mcsta}'s interval generation method in our tool as well, which will allow it to perform better as for the $\mathit{Uni}(0,16)$ variant.

\section{Conclusion and Future Work}
\label{sec:Conclusion}

We have introduced the first model checking technique specifically designed for stochastic automata (SA), based on interval abstraction.
It puts few restrictions on the features that can be used in SA models, and our experiments---though using a prototype implementation---indicate that it can offer much higher scalability than previous techniques for superclasses of~SA.

\paragraph{Future work.}
We plan to turn our prototype into a proper performance-optimised open-source tool connected to the \toolset or \jani.
We can then use the \toolset's Kepler frontend~\cite{DBCH25} to analyse dynamic fault tree models, whose semantics is given in terms of SA (that are usually rather complex, with multiple timers, frequent resets, and hundreds of locations).
To obtain not only upper/lower bounds on max./min.\ probabilities, we can use a game-based approach~\cite{HNPWZ11} to also obtain lower/upper bounds, which could be extended to automatic interval refinement akin to~\cite{KNP09}.
Finally, the presented analysis approach can be extended to deterministic delays and time-bounded properties. %

\paragraph{Data availability.}
An artifact that contains our tool and allows reproducing \Cref{sec:Experiments} is publicly available at DOI \href{https://doi.org/10.5281/zenodo.19829349}{10.5281/zenodo.19829349}~\cite{artifact}. %

\bibliographystyle{splncs04}
\bibliography{paper}

\begin{thebibliography}{10}
\providecommand{\url}[1]{\texttt{#1}}
\providecommand{\urlprefix}{URL }
\providecommand{\doi}[1]{https://doi.org/#1}

\bibitem{AP18}
Agha, G., Palmskog, K.: A survey of statistical model checking. {ACM} Trans.
  Model. Comput. Simul.  \textbf{28}(1),  6:1--6:39 (2018).
  \doi{10.1145/3158668}

\bibitem{BAFK18}
Baier, C., de~Alfaro, L., Forejt, V., Kwiatkowska, M.: Model checking
  probabilistic systems. In: Clarke, E.M., Henzinger, T.A., Veith, H., Bloem,
  R. (eds.) Handbook of Model Checking, pp. 963--999. Springer (2018).
  \doi{10.1007/978-3-319-10575-8_28}

\bibitem{BHHK03}
Baier, C., Haverkort, B.R., Hermanns, H., Katoen, J.: Model-checking algorithms
  for continuous-time {M}arkov chains. {IEEE} Trans. Software Eng.
  \textbf{29}(6),  524--541 (2003). \doi{10.1109/TSE.2003.1205180}

\bibitem{BHHK10}
Baier, C., Haverkort, B.R., Hermanns, H., Katoen, J.P.: Performance evaluation
  and model checking join forces. Commun. {ACM}  \textbf{53}(9),  76--85
  (2010). \doi{10.1145/1810891.1810912}

\bibitem{Bel57}
Bellman, R.: A {M}arkovian decision process. Journal of Mathematics and
  Mechanics  \textbf{6}(5),  679--684 (1957)

\bibitem{BDHK06}
Bohnenkamp, H.C., D'Argenio, P.R., Hermanns, H., Katoen, J.P.: {MoDeST}: A
  compositional modeling formalism for hard and softly timed systems. {IEEE}
  Trans. Software Eng.  \textbf{32}(10),  812--830 (2006).
  \doi{10.1109/TSE.2006.104}

\bibitem{BBD03}
Bryans, J.W., Bowman, H., Derrick, J.: Model checking stochastic automata.
  {ACM} Trans. Comput. Log.  \textbf{4}(4),  452--492 (2003).
  \doi{10.1145/937555.937558}

\bibitem{BKS14}
Buchholz, P., Kriege, J., Scheftelowitsch, D.: Model checking stochastic
  automata for dependability and performance measures. In: 44th Annual
  {IEEE/IFIP} International Conference on Dependable Systems and Networks
  ({DSN} 2014). pp. 503--514. {IEEE} Computer Society (2014).
  \doi{10.1109/DSN.2014.53}

\bibitem{Bud22}
Budde, C.E.: {FIG}: the finite improbability generator v1.3. {SIGMETRICS}
  Perform. Evaluation Rev.  \textbf{49}(4),  59--64 (2022).
  \doi{10.1145/3543146.3543160}

\bibitem{BDHS20}
Budde, C.E., D'Argenio, P.R., Hartmanns, A., Sedwards, S.: An efficient
  statistical model checker for nondeterminism and rare events. Int. J. Softw.
  Tools Technol. Transf.  \textbf{22}(6),  759--780 (2020).
  \doi{10.1007/S10009-020-00563-2}

\bibitem{BDHHJT17}
Budde, C.E., Dehnert, C., Hahn, E.M., Hartmanns, A., Junges, S., Turrini, A.:
  {JANI}: Quantitative model and tool interaction. In: Legay, A., Margaria, T.
  (eds.) 23rd International Conference on Tools and Algorithms for the
  Construction and Analysis of Systems ({TACAS} 2017). Lecture Notes in
  Computer Science, vol. 10206, pp. 151--168. Springer (2017).
  \doi{10.1007/978-3-662-54580-5\_9}

\bibitem{Dar99}
D'Argenio, P.R.: Algebras and Automata for Timed and Stochastic Systems. Ph.D.
  thesis, University of Twente, Netherlands (1999),
  \url{https://research.utwente.nl/en/publications/algebras-and-automata-for-timed-and-stochastic-systems}

\bibitem{DGHS18}
D'Argenio, P.R., Gerhold, M., Hartmanns, A., Sedwards, S.: A hierarchy of
  scheduler classes for stochastic automata. In: Baier, C., Lago, U.D. (eds.)
  21st International Conference on Foundations of Software Science and
  Computation Structures ({FOSSACS} 2018). Lecture Notes in Computer Science,
  vol. 10803, pp. 384--402. Springer (2018).
  \doi{10.1007/978-3-319-89366-2\_21}

\bibitem{DK05}
D'Argenio, P.R., Katoen, J.P.: A theory of stochastic systems part {I}:
  Stochastic automata. Inf. Comput.  \textbf{203}(1),  1--38 (2005).
  \doi{10.1016/J.IC.2005.07.001}

\bibitem{DM18}
D'Argenio, P.R., Monti, R.E.: Input/output stochastic automata with urgency:
  Confluence and weak determinism. In: Fischer, B., Uustalu, T. (eds.) 15th
  International Colloquium on Theoretical Aspects of Computing ({ICTAC} 2018).
  Lecture Notes in Computer Science, vol. 11187, pp. 132--152. Springer (2018).
  \doi{10.1007/978-3-030-02508-3\_8}

\bibitem{DSSR23}
Delicaris, J., St{\"{u}}bbe, J., Schupp, S., Remke, A.: {RealySt}: A {C++} tool
  for optimizing reachability probabilities in stochastic hybrid systems. In:
  Kalyvianaki, E., Paolieri, M. (eds.) 16th {EAI} International Conference on
  Performance Evaluation Methodologies and Tools ({VALUETOOLS} 2023). Lecture
  Notes of the Institute for Computer Sciences, Social Informatics and
  Telecommunications Engineering, vol.~539, pp. 170--182. Springer (2023).
  \doi{10.1007/978-3-031-48885-6\_11}

\bibitem{DBCH25}
Dengler, G., Budde, C.E., Carnevali, L., Hartmanns, A.: Time-sensitive
  importance splitting. In: Prabhakar, P., Vandin, A. (eds.) 2nd International
  Joint Conference on Quantitative Evaluation of Systems and Formal Modeling
  and Analysis of Timed Systems ({QEST+FORMATS} 2025). Lecture Notes in
  Computer Science, vol. 16143, pp. 21--41. Springer (2025).
  \doi{10.1007/978-3-032-05792-1\_2}

\bibitem{EHKZ13}
Eisentraut, C., Hermanns, H., Katoen, J.P., Zhang, L.: A semantics for every
  {GSPN}. In: Colom, J.M., Desel, J. (eds.) 34th International Conference on
  Application and Theory of Petri Nets and Concurrency ({P}etri {N}ets 2013).
  Lecture Notes in Computer Science, vol.~7927, pp. 90--109. Springer (2013).
  \doi{10.1007/978-3-642-38697-8\_6}

\bibitem{FKNPQ11}
Forejt, V., Kwiatkowska, M.Z., Norman, G., Parker, D., Qu, H.: Quantitative
  multi-objective verification for probabilistic systems. In: Abdulla, P.A.,
  Leino, K.R.M. (eds.) 17th International Conference on Tools and Algorithms
  for the Construction and Analysis of Systems ({TACAS} 2011). Lecture Notes in
  Computer Science, vol.~6605, pp. 112--127. Springer (2011).
  \doi{10.1007/978-3-642-19835-9\_11}

\bibitem{FHHWZ11}
Fr{\"{a}}nzle, M., Hahn, E.M., Hermanns, H., Wolovick, N., Zhang, L.:
  Measurability and safety verification for stochastic hybrid systems. In:
  Caccamo, M., Frazzoli, E., Grosu, R. (eds.) 14th {ACM} International
  Conference on Hybrid Systems: Computation and Control ({HSCC} 2011). pp.
  43--52. {ACM} (2011). \doi{10.1145/1967701.1967710}

\bibitem{Fre08}
Frehse, G.: {PHAV}er: algorithmic verification of hybrid systems past
  {H}y{T}ech. Int. J. Softw. Tools Technol. Transf.  \textbf{10}(3),  263--279
  (2008). \doi{10.1007/S10009-007-0062-X}

\bibitem{GRH14}
Ghasemieh, H., Remke, A., Haverkort, B.R.: Hybrid {P}etri nets with multiple
  stochastic transition firings. In: Haviv, M., Knottenbelt, W.J., Maggi, L.,
  Miorandi, D. (eds.) 8th International Conference on Performance Evaluation
  Methodologies and Tools ({VALUETOOLS} 2014). pp. 217--224. {ICST} (2014).
  \doi{10.4108/icst.valuetools.2014.258204}

\bibitem{Hah13}
Hahn, E.M.: Model checking stochastic hybrid systems. Ph.D. thesis, Saarland
  University (2013),
  \url{http://scidok.sulb.uni-saarland.de/volltexte/2013/5259/}

\bibitem{HHH14}
Hahn, E.M., Hartmanns, A., Hermanns, H.: Reachability and reward checking for
  stochastic timed automata. Electron. Commun. Eur. Assoc. Softw. Sci. Technol.
   \textbf{70} (2014). \doi{10.14279/TUJ.ECEASST.70.968}

\bibitem{HHHK13}
Hahn, E.M., Hartmanns, A., Hermanns, H., Katoen, J.P.: A compositional
  modelling and analysis framework for stochastic hybrid systems. Formal
  Methods Syst. Des.  \textbf{43}(2),  191--232 (2013).
  \doi{10.1007/S10703-012-0167-Z}

\bibitem{HNPWZ11}
Hahn, E.M., Norman, G., Parker, D., Wachter, B., Zhang, L.: Game-based
  abstraction and controller synthesis for probabilistic hybrid systems. In:
  8th International Conference on Quantitative Evaluation of Systems ({QEST}
  2011). pp. 69--78. {IEEE} Computer Society (2011). \doi{10.1109/QEST.2011.17}

\bibitem{Har22}
Hartmanns, A.: Correct probabilistic model checking with floating-point
  arithmetic. In: Fisman, D., Rosu, G. (eds.) 28th International Conference on
  Tools and Algorithms for the Construction and Analysis of Systems ({TACAS}
  2022). Lecture Notes in Computer Science, vol. 13244, pp. 41--59. Springer
  (2022). \doi{10.1007/978-3-030-99527-0\_3}

\bibitem{HH14}
Hartmanns, A., Hermanns, H.: The {M}odest {T}oolset: An integrated environment
  for quantitative modelling and verification. In: {\'{A}}brah{\'{a}}m, E.,
  Havelund, K. (eds.) 20th International Conference on Tools and Algorithms for
  the Construction and Analysis of Systems ({TACAS} 2014). Lecture Notes in
  Computer Science, vol.~8413, pp. 593--598. Springer (2014).
  \doi{10.1007/978-3-642-54862-8\_51}

\bibitem{HJQW23}
Hartmanns, A., Junges, S., Quatmann, T., Weininger, M.: A practitioner's guide
  to {MDP} model checking algorithms. In: Sankaranarayanan, S., Sharygina, N.
  (eds.) 29th International Conference on Tools and Algorithms for the
  Construction and Analysis of Systems ({TACAS} 2023). Lecture Notes in
  Computer Science, vol. 13993, pp. 469--488. Springer (2023).
  \doi{10.1007/978-3-031-30823-9_24}

\bibitem{HMP92}
Henzinger, T.A., Manna, Z., Pnueli, A.: What good are digital clocks? In:
  Kuich, W. (ed.) 19th International Colloquium on Automata, Languages and
  Programming ({ICALP} 1992). Lecture Notes in Computer Science, vol.~623, pp.
  545--558. Springer (1992). \doi{10.1007/3-540-55719-9\_103}

\bibitem{How60}
Howard, R.A.: Dynamic Programming and {M}arkov Processes. MIT Press (1960)

\bibitem{KNP09}
Kwiatkowska, M.Z., Norman, G., Parker, D.: Stochastic games for verification of
  probabilistic timed automata. In: Ouaknine, J., Vaandrager, F.W. (eds.) 7th
  International Conference on Formal Modeling and Analysis of Timed Systems
  ({FORMATS} 2009). Lecture Notes in Computer Science, vol.~5813, pp. 212--227.
  Springer (2009). \doi{10.1007/978-3-642-04368-0\_17}

\bibitem{KNP11}
Kwiatkowska, M.Z., Norman, G., Parker, D.: {PRISM} 4.0: Verification of
  probabilistic real-time systems. In: Gopalakrishnan, G., Qadeer, S. (eds.)
  23rd International Conference on Computer Aided Verification ({CAV} 2011).
  Lecture Notes in Computer Science, vol.~6806, pp. 585--591. Springer (2011).
  \doi{10.1007/978-3-642-22110-1\_47}

\bibitem{KNPS06}
Kwiatkowska, M.Z., Norman, G., Parker, D., Sproston, J.: Performance analysis
  of probabilistic timed automata using digital clocks. Formal Methods Syst.
  Des.  \textbf{29}(1),  33--78 (2006). \doi{10.1007/S10703-006-0005-2}

\bibitem{KNSS02}
Kwiatkowska, M.Z., Norman, G., Segala, R., Sproston, J.: Automatic verification
  of real-time systems with discrete probability distributions. Theor. Comput.
  Sci.  \textbf{282}(1),  101--150 (2002). \doi{10.1016/S0304-3975(01)00046-9}

\bibitem{LLTYSG19}
Legay, A., Lukina, A., Traonouez, L.M., Yang, J., Smolka, S.A., Grosu, R.:
  Statistical model checking. In: Steffen, B., Woeginger, G.J. (eds.) Computing
  and Software Science -- State of the Art and Perspectives, Lecture Notes in
  Computer Science, vol. 10000, pp. 478--504. Springer (2019).
  \doi{10.1007/978-3-319-91908-9_23}

\bibitem{MCB84}
Marsan, M.A., Conte, G., Balbo, G.: A class of generalized stochastic {P}etri
  nets for the performance evaluation of multiprocessor systems. {ACM} Trans.
  Comput. Syst.  \textbf{2}(2),  93--122 (1984). \doi{10.1145/190.191}

\bibitem{artifact}
Petri, A.: Effective stochastic automata model checking by interval abstraction
  (artifact). Zenodo (2026). \doi{10.5281/zenodo.19829349}

\bibitem{Spr00}
Sproston, J.: Decidable model checking of probabilistic hybrid automata. In:
  Joseph, M. (ed.) 6th International Symposium on Formal Techniques in
  Real-Time and Fault-Tolerant Systems ({FTRTFT} 2000). Lecture Notes in
  Computer Science, vol.~1926, pp. 31--45. Springer (2000).
  \doi{10.1007/3-540-45352-0\_5}

\bibitem{Wol12}
Wolovick, N.: Continuous probability and nondeterminism in labeled transition
  systems. Ph.D. thesis, Universidad Nacional de C{\'o}rdoba, C{\'o}rdoba,
  Argentina (2012)

\bibitem{ZSRHH10}
Zhang, L., She, Z., Ratschan, S., Hermanns, H., Hahn, E.M.: Safety verification
  for probabilistic hybrid systems. In: Touili, T., Cook, B., Jackson, P.B.
  (eds.) 22nd International Conference on Computer Aided Verification ({CAV}
  2010). Lecture Notes in Computer Science, vol.~6174, pp. 196--211. Springer
  (2010). \doi{10.1007/978-3-642-14295-6\_21}

\end{thebibliography}

\clearpage
\appendix

\section{Stochastic Automata from the Literature}
\label{app:FossacsAutomata}

\Cref{fig:Cex3} shows the three small SA $M_1$, $M_2$, and $M_5$ from~\cite{DGHS18} for which \Cref{tab:results_small_automata} shows results in \Cref{sec:ExperimentalResults}. In this section we show the exact computation of maximal and minimal reachability probabilities.

\begin{figure}[b]
\begin{minipage}[b]{0.33\textwidth}
\centering
\begin{tikzpicture}[on grid,auto]
  \node[state] (l0) {$\ell_0$};
  \coordinate[left=0.3 of l0.west] (start);
  \node[] (me) [above left=0.4 and 1.1 of l0] {\small$M_1$:};
  \node[] (distr) [above right=0.2 and 1.3 of l0,align=left] {$x\colon \textsc{Uni}(0, 1)$\\$y\colon \textsc{Uni}(0, 1)$};
  \node[state] (l1) [below=1 of l0] {$\ell_1$};
  \node[state] (l2) [below left=0.875 and 0.75 of l1] {$\ell_2$};
  \node[state] (l3) [below right=0.875 and 0.75 of l1] {$\ell_3$};
  \node[state] (yes) [below=1 of l2] {\cmark};
  \node[state] (no) [below=1 of l3] {\xmark};
  ;
  \path[->]
    (start) edge node {} (l0)
    (l0) edge node[right,pos=0.25,inner sep=0.5mm] {\strut$\varnothing$} node[right,pos=0.7,inner sep=0.5mm] {\strut$\text{\restart}(\{ x \})$} (l1)
    (l1) edge[] node[left,pos=0.15,inner sep=1mm] {\strut$\varnothing\!$} node[left,pos=0.55,inner sep=1mm] {\strut$\text{\restart}(\{ y \})$} (l2)
    (l1) edge[] node[right,pos=0.15,inner sep=1mm] {\strut$\varnothing$} node[right,pos=0.55,inner sep=1mm] {\strut$\text{\restart}(\{ y \})$} (l3)
    (l2) edge node[left,pos=0.25,inner sep=0.5mm] {\strut$\{ x \}$} (yes)
    (l2) edge[bend right=20] node[above,pos=0.33,inner sep=2mm] {\strut$\{ y \}$} (no)
    (l3) edge[bend left=20] node[above,pos=0.33,inner sep=2mm] {\strut$\{ y \}$} (yes)
    (l3) edge node[right,pos=0.25,inner sep=0.5mm] {\strut$\{ x \}$} (no)
  ;
\end{tikzpicture}
\caption{SA $M_1$}
\label{fig:Cex2}
\end{minipage}%
\begin{minipage}[b]{0.33\textwidth}
\centering
\begin{tikzpicture}[on grid,auto]
  \node[state] (l0) {$\ell_0$};
  \coordinate[left=0.3 of l0.west] (start);
  \node[] (me) [above left=0.4 and 1.1 of l0] {\small$M_2$:};
  \node[] (distr) [above right=0.2 and 1.3 of l0,align=left] {$x\colon \textsc{Uni}(0, 8)$\\$y\colon \textsc{Uni}(0, 1)$\\$z\colon \textsc{Uni}(0, 4)$};
  \node[state] (l1) [below=1 of l0] {$\ell_1$};
  \node[state] (l2) [below=1 of l1] {$\ell_2$};
  \node[state] (l3) [below left=0.875 and 0.75 of l2] {$\ell_3$};
  \node[state] (l4) [below right=0.875 and 0.75 of l2] {$\ell_4$};
  \node[state] (yes) [below=1 of l3] {\cmark};
  \node[state] (no) [below=1 of l4] {\xmark};
  ;
  \path[->]
    (start) edge node {} (l0)
    (l0) edge node[right,pos=0.25,inner sep=0.5mm] {\strut$\varnothing$} node[right,pos=0.7,inner sep=0.5mm] {\strut$\text{\restart}(\{ x, z \})$} (l1)
    (l1) edge[bend right=50] node[left,pos=0.25,inner sep=0.5mm] {\strut$\{ x \}$} node[left,pos=0.7,inner sep=0.5mm] {\strut$\text{\restart}(\{z \})$} (l2)
    (l1) edge[bend left=50] node[right,pos=0.25,inner sep=0.5mm] {\strut$\{z \}$} node[right,pos=0.7,inner sep=0.75mm] {\strut$\text{\restart}(\{z \})$} (l2)
    (l2) edge[] node[left,pos=0.15,inner sep=1mm] {\strut$\varnothing\!$} node[left,pos=0.55,inner sep=1mm] {\strut$\text{\restart}(\{ y \})$} (l3)
    (l2) edge[] node[right,pos=0.15,inner sep=1mm] {\strut$\varnothing$} node[right,pos=0.55,inner sep=1mm] {\strut$\text{\restart}(\{ y \})$} (l4)
    (l3) edge node[left,pos=0.25,inner sep=0.5mm] {\strut$\{ x \}$} (yes)
    (l3) edge[bend right=20] node[above,pos=0.33,inner sep=2mm] {\strut$\{ y \}$} (no)
    (l4) edge[bend left=20] node[above,pos=0.33,inner sep=2mm] {\strut$\{ y \}$} (yes)
    (l4) edge node[right,pos=0.25,inner sep=0.5mm] {\strut$\{ x \}$} (no)
  ;
\end{tikzpicture}
\caption{SA $M_2$}
\label{fig:CexX}
\end{minipage}%
\begin{minipage}[b]{0.33\textwidth}
\centering
\begin{tikzpicture}[on grid,auto]
 \node[state] (l0) {$\ell_1$};
  \coordinate[left=0.3 of l0.west] (start);
  \node[] (me) [above left=0.4 and 1.1 of l0] {\small$M_5$:};
  \node[] (distr) [above right=0.2 and 1.3 of l0,align=left] {$x\colon \textsc{Uni}(0, 1)$\\$y\colon \textsc{Uni}(0, 1)$};
  \node[state] (l1) [below=1 of l0] {$\ell_2$};
  \node[state] (l2) [below=1 of l1] {$\ell_3$};
  \node[state] (l3) [below left=0.875 and 0.75 of l2] {$\ell_4$};
  \node[state] (l4) [below right=0.875 and 0.75 of l2] {$\ell_5$};
  \node[state] (yes) [below=0.84 of l3] {\cmark};
  \node[state] (no) [below=0.84 of l4] {\xmark};
  ;
  \path[->]
    (start) edge node {} (l0)
    (l0) edge node[right,pos=0.25,inner sep=0.5mm] {\strut$\varnothing$} node[right,pos=0.7,inner sep=0.5mm] {\strut$\text{\restart}(\{ x, y \})$} (l1)
    (l1) edge node[right,pos=0.25,inner sep=0.5mm] {\strut$\{ y \}$} node[right,pos=0.7,inner sep=0.5mm] {\strut$\text{\restart}(\{ y \})$} (l2)
    (l2) edge[] node[left,pos=0.15,inner sep=1mm] {\strut$\varnothing\!$} node[left,pos=0.55,inner sep=1mm] {\strut$\text{\restart}(\{ y \})$} (l3)
    (l2) edge[] node[right,pos=0.15,inner sep=1mm] {\strut$\varnothing$} node[right,pos=0.55,inner sep=1mm] {\strut$\text{\restart}(\{ y \})$} (l4)
    (l3) edge node[left,pos=0.4,inner sep=0.5mm] {\strut$\{ x \}$} (yes)
    (l3) edge[bend right=20] node[above,pos=0.33,inner sep=2mm] {\strut$\{ y \}$} (no)
    (l4) edge[bend left=20] node[above,pos=0.33,inner sep=2mm] {\strut$\{ y \}$} (yes)
    (l4) edge node[right,pos=0.4,inner sep=0.5mm] {\strut$\{ x \}$} (no)
  ;
\end{tikzpicture}
\caption{SA $M_5$}
\label{fig:Cex3}
\end{minipage}%
\end{figure}

\subsection{Exact Probabilities for $M_1$}
The exact probabilities for reaching the goal in $M_1$ are $\pmin = 0.25$ and $\pmax = 0.75$ and can be computed as

\begin{align*}
\textstyle\pmax &= P(X\leq 0.5) \cdot P(X\leq Y) + P(X\geq 0.5) \cdot P(X\geq Y)\\
&= 0.5\cdot(\int_{0}^{0.5} \int_{x}^{1} f(y) \,dy f(x) \,dx) + 0.5\cdot(\int_{0.5}^{1} \int_{0}^{x} f(y) \,dy f(x)\,dx)\\ &= 0.5 \cdot \int_{0}^{0.5} 2(1-x) \, dx + 0.5\cdot\int_{0.5}^{1} 2x\,dx\\ &= 0.75\\
\end{align*}
and 
\begin{align*}
\textstyle\pmin &= P(X\leq 0.5) \cdot P(X\geq Y) + P(X\geq 0.5) \cdot P(X\leq Y)\\
&= 0.5\cdot(\int_{0}^{0.5} \int_{0}^{x} f(y) \,dy f(x) \,dx) + 0.5\cdot(\int_{0.5}^{1} \int_{x}^{1} f(y) \,dy f(x)\,dx)\\ &= 0.5 \cdot \int_{0}^{0.5} 2x \, dx + 0.5\cdot\int_{0.5}^{1} 2(1-x)\,dx\\ &= 0.25.
\end{align*}

The only nondeterministic choice present in $M_1$ is the choice of the outgoing transition in $\ell_1$. If the transition to $\ell_2$ is chosen $x$ must expire before $y$ to reach the target state, if the transition to $\ell_3$ is chosen $y$ must expire before $x$ instead. Since $X,Y \sim Uni(0,1)$ and while the valuation of $x$ is known in $\ell_1$, the valuation of $y$ is resampled on the transitions. Hence, for maximizing the probability of reaching the target location, going to $\ell_2$ is optimal if $\vec{x} \leq 0.5$ which is represented by the first integral, where the outer integral for $x$ ranges from 0 to 0.5 and the inner integral for $y$ ensures that $y$ expires after $x$. In the opposite case where $\vec{x} \geq 0.5$ going to $\ell_3$ is optimal, which is accounted for by the second integral.

\subsection{Exact Probabilities for $M_2$}
$M_2$ has a similar nondeterministic choice in location $\ell_2$. However, the valuation of $x$ in that location depends now also on an additional timer variable $z$. If initially the sampled valuation $x$ is lower than the sampled valuation of $z$, the transition to $\ell_2$ guarded by $x$ is taken, meaning that $x = \expired$ in $\ell_2$ and thus allowing optimal decisions. In the opposite case where initially $\vec{x} > \vec{z}$ the optimal choice depends on the remaining valuation of $x$ in location $\ell_2$. The CPD for that can be computed as $C=X-Z$.

\subsubsection{CDF of $C = X - Z$.}
Since $X \sim \mathrm{Uniform}[0,8]$ and $Y \sim \mathrm{Uniform}[0,4]$, so their PDFs are
\[
  f_X(x) = \frac{1}{8} \quad (x \in [0,8]), \qquad f_Y(y) = \frac{1}{4} \quad (y \in [0,4]).
\]
The difference $Z = X - Y$ ranges over $[-4,\, 8]$.
The combined PDF $f_Z(z)$ can be obtained using convolution
 
\[
  f_Z(z)
  = \int_{-\infty}^{\infty} f_X(t)\, f_Y(t - z)\, dt
  = \int_{-\infty}^{\infty}
      \frac{1}{8} \cdot \frac{1}{4}
      \cdot \mathbf{1}[t \in [0,8]]
      \cdot \mathbf{1}[t - z \in [0,4]]\, dt.
\]
 
The second indicator requires $t \in [z,\, z + 4]$, so the effective integration range is
$t \in [\max(0,\, z),\, \min(8,\, z+4)]$, and the integrand equals $\tfrac{1}{32}$ throughout.
This interval changes across three cases.
\begin{itemize}
    \item {Case 1: $-4 \leq z < 0$.}
Both constraints give $t \geq 0$ and $t \leq z + 4$. Since $z + 4 < 4 \leq 8$ we have
$\min(8,\, z+4) = z + 4$, so the length is $(z+4) - 0 = z + 4$:
\[
  f_Z(z) = \frac{1}{32}(z + 4).
\]
\item {Case 2: $0 \leq z \leq 4$.}
Lower bound: $\max(0,z) = z$. Upper bound: $\min(8,\, z+4) = z+4$ (since $z + 4 \leq 8$).
Length $= 4$:
\[
  f_Z(z) = \frac{4}{32} = \frac{1}{8}.
\]
\item {Case 3: $4 < z \leq 8$.}
Lower bound: $\max(0,z) = z$. Upper bound: $\min(8,\, z+4) = 8$ (since $z + 4 > 8$).
Length $= 8 - z$:
\[
  f_Z(z) = \frac{1}{32}(8 - z).
\]
\end{itemize}

Combining all three cases gives the full PDF:
\[
  \boxed{f_Z(z) = \begin{cases}
    \dfrac{z+4}{32}  & -4 \leq z < 0, \\[8pt]
    \dfrac{1}{8}     &  0 \leq z \leq 4, \\[8pt]
    \dfrac{8-z}{32}  &  4 < z \leq 8.
  \end{cases}}
\]
 
Using the combined probability distribution for $C = X -Z$ the exact maximal reachability probability can be computed as: 
\begin{align*}
\textstyle\pmax &= P(X-Z \leq 0) + P(X-Z \leq 0.5 \wedge X-Z \leq Y)\\&\phantom{=}\ + P(X-Z \geq 0.5 \wedge X-Z \geq Y)\\
&= \int_{-4}^{0} \frac{1}{32} (c + 4) \,dc + \int_{0}^{0.5} \int_{c}^{1} f_C(c) f_Y(y) \,dc \,dy \\&\phantom{=}\ + \int_{0.5}^{8} \int_{0}^{c} f_C(c) f_Y(y) \,dc \,dy\\
&= \int_{-4}^{0} \frac{1}{32} (c + 4) \,dc + \int_{0}^{0.5} \int_{c}^{1} f_C(c) f_Y(y) \,dc \,dy\\&\phantom{=}\ + \int_{0.5}^{1} \int_{0}^{c} f_C(c) f_Y(y) \,dc \,dy + \int_{1}^{8} f_C(c)\,dc \\ 
&= \int_{-4}^{0} \frac{1}{32} (c + 4) \,dc + \int_{0}^{0.5} \int_{c}^{1} \frac{1}{8} \cdot 1 \,dc \,dy\\&\phantom{=}\ + \int_{0.5}^{1} \int_{0}^{c} \frac{1}{8} \cdot 1 \,dc \,dy + \int_{1}^{4} f_C(c)\,dc + \int_{4}^{8} f_C(c)\,dc\\ 
&= \int_{-4}^{0} \frac{1}{32} (c + 4) \,dc + \int_{0}^{0.5} \frac{1}{8} (1-c) \,dc\\&\phantom{=}\ + \int_{0.5}^{1} \frac{1}{8} c \,dc + \int_{1}^{4} f_C(c)\,dc + \int_{4}^{8} f_C(c)\,dc\\ 
&= \frac{1}{32} \left[0.5c^2 + 4c \right]_{-4}^{0} + \left[\frac{1}{8}c - \frac{1}{8} \cdot \frac{1}{2} c^2 \right]_{0}^{0.5} + \left[ \frac{1}{16} c^2 \right]_{0.5}^{1}\\&\phantom{=}\ + \left[ \frac{1}{8} c \right]_{1}^{4} + \frac{1}{32}\left[ 8 c - \frac{1}{2} c^2\right]_{4}^{8}\\ 
&= \frac{1}{4} + (\frac{1}{8}(\frac{1}{2} - \frac{1}{8}) + (\frac{1}{16} - \frac{1}{16} \cdot \frac{1}{4}) + (\frac{4}{8} - \frac{1}{8}) + \frac{1}{32}\left( \frac{64}{2} - 24 \right)\\
&= \frac{1}{4} + \frac{1}{8} \cdot\frac{3}{8} + \frac{3}{64} + \frac{3}{8} + \frac{8}{32}\\
&= \frac{31}{32} = 0.96875\\
\end{align*}

and, since for $C < 0$ and $C >1$ an optimal decision which avoids the target state is possible, 

\begin{align*}
\textstyle\pmin &= P(0 \leq C \leq 0.5 \wedge Y \leq C) + P(0.5 \leq C \leq 1 \wedge Y \geq C)\\
&= \int_{0}^{0.5} \int_{0}^{c} f_C(c) \,dc f_Y(y) \,dy + \int_{0.5}^{1} \int_{c}^{1} f_C(c) \,dc f_Y(y) \,dy\\ 
&= \int_{0}^{0.5} \frac{1}{8} c \,dc + \int_{0.5}^{1} \frac{1}{8} (1-c) \,dc \\ 
&= \left[ \frac{1}{16} c^2 \right]_{0}^{0.5} + \frac{1}{8} \left[c - 0.5 c^2 \right]_{0.5}^{1} \\ 
&= \frac{1}{16} \cdot \frac{1}{4} + \frac{1}{8} (1 - \frac{1}{2} - (\frac{1}{2} - \frac{1}{8}))\\
&= \frac{1}{16} \cdot \frac{1}{4} + \frac{1}{8} (\frac{1}{8})\\
&= \frac{1}{32} = 0.03123\\.
\end{align*}

\subsection{Exact Probabilities for $M_5$}
For $M_5$ the optimal probabilities to reach or avoid the target state are determined by the nondeterministic choice in location $\ell_3$. If initially the sampled valuation for $x$ is below the sampled valuation of $y$, then x is already expired in $\ell_3$. Since both both timers are uniformly distributed on $(0,1)$ the probability that is $P(X \leq Y) = 0.5 $ and in that case $x= \expired$ means that the reset value of $y$ does not influence the result.
In the opposite case where initially $\vec{v} > \vec{y}$ the optimal choice in Location $\ell_3$ depends on the remaining valuation of $x$. The optimal decisions are similarly to $M_1$: if $\vec{x} \leq 0.5$ the optimal choice for maximal winning chances is taking the transition to Location $\ell_4$ and if $v(x) \geq 0.5$ it is the transition to $\ell_5$. The remaining valuation of $x$ can be computed as the joint probability function $Z = X-Y$ .

\subsubsection{PDF of $Z = X - Y$}
Since $X, Y \sim Uni(0,1)$ independently, $Z$ ranges over $[-1, 1]$. Using convolution, the probability density function of $Z$ is:
\[
  f_Z(z) = \int_{-\infty}^{\infty} f_X(t)\, f_Y(t - z)\, dt.
\]
Both densities equal $1$ on $[0,1]$ and $0$ elsewhere, so the integrand is $1$ only when
$t \in [0,1]$ and $t - z \in [0,1]$, i.e.\ $t \in [\max(0,z),\, \min(1, 1+z)]$.

\begin{itemize}
    \item {Case 1: $0 \leq z \leq 1$.}
\[
  f_Z(z) = \int_z^1 dt = 1 - z.
\]
\item {Case 2: $-1 \leq z < 0$.}
\[
  f_Z(z) = \int_0^{1+z} dt = 1 + z.
\]
\end{itemize}

Hence the full PDF is:
\[
  \boxed{f_Z(z) = 1 - |z|, \qquad z \in [-1,\, 1].}
\]
This is the \textit{symmetric triangular distribution} on $[-1, 1]$, with the peak at $z = 0$.
The exact probabilities for reaching the goal in $M_5$ are $\pmin = 0.125$ and $\pmax = 0.875$ and can be computed as

\begin{align*}
p_{\max} &= P(X \leq Y_0) + P(0\leq X-Y_0 \leq 0.5) + P(X-Y_0 \geq 0.5)\\
&= 0.5 + 0.5\cdot\left(\int_{0}^{0.5} \int_{z}^{1} f_z(z) \,dz\, f_y(y) \,dy\right) + 0.5\cdot\left(\int_{0.5}^{1} \int_{0}^{z} f_z(z) \,dz\, f_y(y)\,dy\right) \\
         &= 0.5 + 0.5\cdot\left(\int_{0}^{0.5} f_z(z)(1-z) \,dz\right) + 0.5\cdot\left(\int_{0.5}^{1} f_z(z)\,z \,dz\right) \\
         &= 0.5 + 0.5\cdot\left(\int_{0}^{0.5} (1-z)^2 \,dz\right) + 0.5\cdot\left(\int_{0.5}^{1} z(1-z) \,dz\right) \\
         &= 0.5 + \left[z - z^2 + \frac{1}{3}z^3\right]_{0}^{0.5} + \left[\frac{1}{2}z^2 - \frac{1}{3}z^3\right]_{0.5}^{1} \\
         &= 0.5 + \left(0.5 - 0.25 + \frac{1}{3}\cdot\frac{1}{8}\right) - 0 + \left(\frac{1}{2} - \frac{1}{3}\right) - \left(\frac{1}{8} - \frac{1}{3}\cdot\frac{1}{8}\right) \\
         &= 0.875
\end{align*}
and
\begin{align*}
p_{\min} &= 0.5\cdot\left(\int_{0}^{0.5} \int_{0}^{z} f_z(z) \,dz\, f_y(y) \,dy\right) + 0.5\cdot\left(\int_{0.5}^{1} \int_{z}^{1} f_z(z) \,dz\, f_y(y)\,dy\right) \\
         &= 0.5\cdot\left(\int_{0}^{0.5} f_z(z)z \,dz\right) + 0.5\cdot\left(\int_{0.5}^{1} f_z(z)(1-z) \,dz\right) \\
         &= 0.5\cdot\left(\int_{0}^{0.5} (1-z)z \,dz\right) + 0.5\cdot\left(\int_{0.5}^{1} (1-z)^2 \,dz\right) \\
         &=  \left[\frac{1}{2}z^2 - \frac{1}{3}z^3 \right]_{0}^{0.5} + \left[ z - z^2 + \frac{1}{3}z^3 \right]_{0.5}^{1} \\
         &= \frac{1}{8} - \frac{1}{3}\cdot\frac{1}{8} - 0 + 1 -1 + \frac{1}{3} -\left(\frac{1}{2} - \frac{1}{2} + \frac{1}{3}\cdot\frac{1}{8}\right) \\
         &= 0.125
\end{align*}

\end{document}